\documentclass[lettersize,journal]{IEEEtran}
\usepackage{amsmath,amsfonts}
\usepackage{array}
\usepackage{textcomp}
\usepackage{stfloats}
\usepackage{url}
\usepackage{verbatim}
\usepackage{graphicx}
\usepackage{cite}
\usepackage{pifont}
\usepackage{makecell}
\usepackage{bm}
\usepackage{bbm}
\usepackage[english]{babel}
\usepackage{amsthm}
\usepackage{color}
\usepackage{soul}

\usepackage{amssymb}
 \usepackage{accents}
  \usepackage[ruled,linesnumbered,lined]{algorithm2e}
  \usepackage{algorithmic}
    \usepackage{framed} 
    \usepackage{comment}
        \usepackage{mathtools}
        \usepackage{hyperref}
        \usepackage{subfigure}

\DeclareMathOperator*{\argmax}{arg\,max}

\newtheorem{theorem}{Theorem}
\newtheorem{definition}{Definition}

\newtheorem{assumption}{Assumption}
\newtheorem{lemma}{Lemma}
\newtheorem{corollary}{Corollary}

\begin{document}

\title{Multi-Fidelity Bayesian Optimization for Nash \\Equilibria with Black-Box Utilities}

\author{Yunchuan Zhang, \IEEEmembership{Member,~IEEE}, Osvaldo Simeone, \IEEEmembership{Fellow,~IEEE}, \\and~H. Vincent Poor, \IEEEmembership{Life Fellow,~IEEE}
\thanks{ 
This work was supported in part by European Union’s Horizon Europe project CENTRIC (101096379), by the Open Fellowships of the EPSRC (EP/W024101/1), by the EPSRC project (EP/X011852/1), and by U.S. National Science Foundation under Grant ECCS-2335876.}

\thanks{Yunchuan Zhang is with National Engineering Research Center of Fiber Optic Sensing Technology and Networks, School of Information Engineering, Hubei Key Laboratory of Broadband Wireless Communication 
and Sensor Networks, Wuhan University of Technology, Wuhan 430070, China (e-mail: yunchuan.zhang@whut.edu.cn)}

\thanks{Osvaldo Simeone is with the King’s Communications, Learning and Information Processing (KCLIP) lab within the Centre for Intelligent Information Processing Systems (CIIPS), Department of Engineering, King’s College London, London, WC2R 2LS, UK. (email: osvaldo.simeone@kcl.ac.uk).}

\thanks{H. Vincent Poor is with the Department of Electrical and Computer Engineering, Princeton University, Princeton, NJ 08544 USA (e-mail: poor@princeton.edu).
}
}



\maketitle

\begin{abstract}
Modern open and softwarized systems -- such as O-RAN telecom networks and cloud computing platforms -- host independently developed applications with distinct, and potentially conflicting, objectives. Coordinating the behavior of such applications to ensure stable system operation poses significant challenges, especially when each application's utility is accessible only via costly, black-box evaluations. In this paper, we consider a centralized optimization framework in which a system controller suggests joint configurations to multiple strategic players, representing different applications, with the goal of aligning their incentives toward a stable outcome. This interaction is modeled as a learned optimization with an equilibrium constraint in which the central optimizer learns the utility functions through sequential, multi-fidelity evaluations with the goal of identifying a pure Nash equilibrium (PNE). To address this challenge, we propose MF-UCB-PNE, a novel multi-fidelity Bayesian optimization strategy that leverages a budget-constrained sampling process to approximate PNE solutions. MF-UCB-PNE systematically balances exploration across low-cost approximations with high-fidelity exploitation steps, enabling efficient convergence to incentive-compatible configurations. We provide theoretical and empirical insights into the trade-offs between query cost and equilibrium accuracy, demonstrating the effectiveness of MF-UCB-PNE in identifying effective equilibrium solutions under limited cost budgets.
\end{abstract}


\section{Introduction}\label{sec: intro}
\subsection{Context and Motivation}\label{ssec: context and motivation}
\IEEEPARstart{I}{n} open, softwarized, systems, different applications are run simultaneously on shared physical resources while pursuing separate, and potentially \emph{conflicting}, objectives. As illustrated in Fig. \ref{fig: oran example}, an important example is given by open telecom network architectures such as O-RAN \cite{polese2023understanding}, in which apps implementing network functionalities are designed to optimize distinct objectives by controlling physical and virtual parameters such as transmission powers and computational resources \cite{giannopoulos2025comix,zhang2023learning}. Examples include cloud computing environments involving tenants with diverse service level agreements \cite{mondal2023fairness}, and scheduling policy selection for service slices \cite{polese2022colo}.

\begin{figure}[t]

  \centering
  \centerline{\includegraphics[scale=0.4]{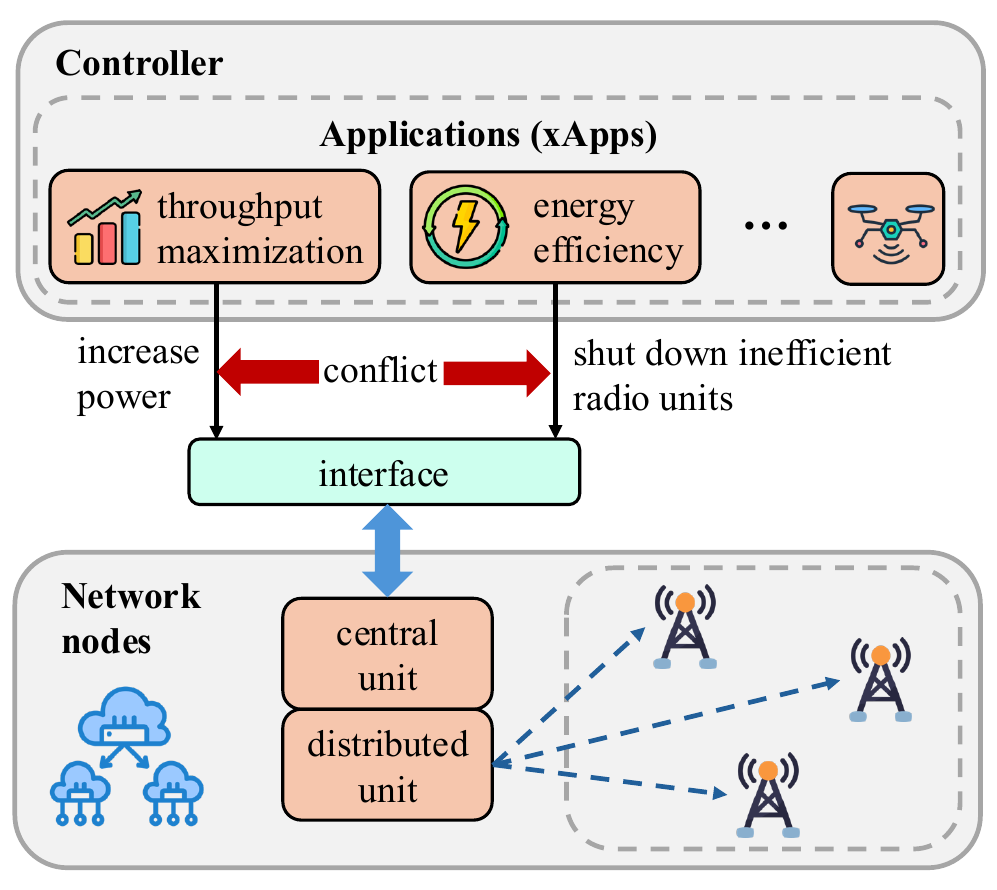}}
 \vspace{-0.2cm}
  \caption{An example of settings captured by the framework studied in this paper: In an open system such as an O-RAN architecture, applications, known as xApps, are deployed simultaneously at a controller. The xApps have different utility functions, and thus their different preferences may cause a conflict. For example, an xApp targeting energy efficiency may  attempt to shut down inefficient radio units, while other xApps may attempt to maximize throughput, causing a conflict on the use of the available transmission resources. The controller can seek joint configurations for the two apps that provide limited incentives for the two apps to deviate from it. }
  \label{fig: oran example}
\vspace{-0.3cm}
\end{figure}

The potentially conflicting nature of the objectives of different apps may yield instabilities and undesirable behaviors \cite{giannopoulos2025comix}. For example, in the setting of Fig. \ref{fig: oran example}, one app aims at controlling energy consumption by adaptively shutting down a transmitter, while another attempts to maximize the communication throughput by increasing energy consumption. This may yield reduced performance levels when both apps are deployed simultaneously. Other examples include scheduling apps for the medium access control (MAC) layer, with different apps focusing on the throughput optimization of distinct users \cite{chu2015prime}.

Motivated by the management of conflicts between independent apps in complex systems like telecom networks, this paper studies the setting depicted in Fig. \ref{fig: intro flow}. In it, a \emph{central optimizer} suggests actions to different \emph{players}, representing the apps in the motivating examples above. The goal for the optimizer is to ensure that the players have no, or limited, incentives to deviate from the suggested actions. Such set of actions is thus compatible with the individual incentives of the players. This ensures a stable operation of the system as the players adhere to the suggestions of the central optimizer.

This type of setting may be modeled as a \emph{learned optimization problem with equilibrium constraints}, in which the actions of competing players are, at least partially, controlled by a common optimizer \cite{pang2013joint,liang2015waveform,yu2017distributed}. In conventional optimization with equilibrium constraints, the utilities of the players are assumed to be \emph{known} to the centralized optimizer \cite{luo1996mathematical}. This is not the case where the optimizer needs to control black-box apps designed by third parties \cite{zhang2021finding}. For instance, in the O-RAN architectures, apps are ideally designed and provided by a variety of vendors, diversifying the supply chain for telecom networks \cite{aijaz2023open}. Therefore, in such settings, the utility functions must be modeled as \emph{black-box functions}.

In order to identify equilibrium solutions for black-box utility functions, the centralized optimizer may operate \emph{sequentially}, exploring alternatives and observing the utility values accrued by the different players. This exploration phase generally entails a \emph{cost} as utility evaluations may entail running real-world experiments or simulations based on, e.g., digital twins of the system \cite{chen2024neuromorphic,huang2023digital,kapteyn2021probabilistic}. To account for this, as also shown in Fig. \ref{fig: intro flow}, we allow the central optimizer to evaluate candidate solutions at different \emph{fidelity} levels, with higher fidelity levels requiring a higher cost. The problem of interest is that of designing centralized optimization strategies that balance the need to obtain good solutions with desirable equilibrium properties with the exploration cost. 

\subsection{Related Work}\label{ssec: related works}
This work lies at the intersection of several research domains, notably black-box optimization, under both single-fidelity \cite{frazier2018tutorial,wang2022tight,srinivas2012information,jones1998efficient,frazier2008knowledge} and multi-fidelity \cite{song2019general,zhang2025multi,kandasamy2017multi} data acquisition settings, and the evaluation of Nash equilibria, considering both known \cite{xu2024consistency,cao2025game,dreves2011solution,shanbhag2011complementarity,daskalakis2009complexity} and unknown \cite{fearnley2015learning,tay2023no,han2024no} utility functions. We then discuss each of these research areas in turn.

\emph{Single-objective black-box optimization}: For expensive-to-evaluate black-box optimization problems, \emph{Bayesian optimization} (BO) \cite{frazier2018tutorial,wang2022tight} provides a common sample-efficient framework. BO leverages probabilistic surrogates, typically Gaussian Processes (GPs), to model the underlying optimization function. It proceeds sequentially, by querying new candidate solutions via an uncertainty-aware acquisition function based on the surrogate function such as upper confidence bound (UCB) \cite{srinivas2012information}, expected improvement (EI) \cite{jones1998efficient}, and knowledge gradient (KG) \cite{frazier2008knowledge}.

\emph{Multi-fidelity single-objective black-box optimization}: Prior works has leveraged multi-fidelity information sources to reduce the query costs. Specifically, multi-fidelity BO (MFBO) builds on multi-fidelity surrogate models, enabling cost-aware acquisition functions that adaptively prioritize evaluations across fidelity levels. References \cite{song2019general,zhang2025multi} proposed multi-fidelity information-theoretic acquisition functions to guide the selection of iterates and provided formal theoretical guarantees on the regret performance; while the work \cite{kandasamy2017multi} extended UCB to the multi-fidelity optimization setting.

\emph{Centralized Nash equilibrium evaluation with known utility functions}: Computing Nash equilibria (NE) in strategic games with \emph{known} utility functions is a common task in many engineering applications \cite{xu2024consistency,cao2025game}. Typical approaches with theoretical guarantees often rely on mathematical programming, such as solving nonlinear equations derived from equilibrium conditions \cite{dreves2011solution} or leveraging complementarity problems \cite{shanbhag2011complementarity}. Alternative paradigms, such as fixed-point iteration algorithms or homotopy continuation methods, provide tractable solutions for polynomial utility structures but encounter instability in non-smooth or high-degree systems \cite{daskalakis2009complexity}.

\emph{Centralized Nash equilibrium evaluation with unknown utility functions}: When analytical expressions of the utility functions are \emph{unknown}, the centralized evaluation of a NE must rely on observations of the utilities for involved players. Reference \cite{fearnley2015learning} characterized the query complexity of finding a NE in various game classes; while references \cite{tay2023no,han2024no} introduced BO-based optimizers that can approximate NE with formal regret guarantees. However, prior works do not account for the possibility to leverage multi-fidelity observations.

\subsection{Main Contributions}\label{ssec: main contributions}
Motivated by modern softwarized systems such as open telecom networks (see Fig. \ref{fig: oran example}), this paper addresses the setting in Fig. \ref{fig: intro flow}, in which a centralized optimizer recommends a joint configuration for the actions of a number of players so that the players have limited incentives to deviate from the suggested actions. To this end, we assume that the centralized optimizer uses \emph{multi-fidelity} estimates of the expensive-to-evaluate black-box utility functions of the players with ordered query costs. 

In this context, we introduce MF-UCB-PNE, a novel MFBO policy that, under a limited query budget,  aims at approaching a solution that approximately satisfies \emph{pure Nash equilibrium} (PNE) conditions. MF-UCB-PNE operates in episodes that encompass \emph{exploration} -- to acquire information about the utility functions from low fidelity observations -- as well as \emph{exploitation} -- to evaluate the current solutions using maximum-fidelity observations.  The main contributions are as follows:
\begin{itemize}
    \item We introduce MF-UCB-PNE, a novel multi-fidelity policy targeting the evaluation of approximate PNEs in general-sum games with black-box utility functions. MF-UCB-PNE generalizes UCB-PNE \cite{tay2023no} to a multi-fidelity framework. MF-UCB-PNE adaptively switches between exploring the utility functions using low-fidelity estimates and exploiting the collected information to evaluate  the quality of the current solution. Adaptation is guided by an information gain per unit cost criterion \cite{song2019general}.  
    \item We provide the \emph{regret} analysis of the proposed MF-UCB-PNE policy with respect to the utility values attainable with optimal approximate PNE solutions. The theoretical results show that MF-UCB-PNE ensures asymptotic \emph{no-regret} performance given sufficiently large query budget.
    \item We present experimental results across synthetic games and wireless communications applications. The results reveal that MF-UCB-PNE produces higher-quality approximate PNE solutions within query budgets compared to single-fidelity policies.
    
\end{itemize}

\subsection{Organization}
The rest of the paper is organized as follows. Sec. \ref{sec: pf} formulates the multi-fidelity general-sum game with black-box utility functions. The proposed MF-UCB-PNE is introduced in Sec. \ref{sec: mfbo for pne}. A theoretical analysis of MF-UCB-PNE is provided in Sec. \ref{sec: theorems}. Numerical results  are provided in Sec. \ref{sec: experiment}. Finally, Sec. \ref{sec: conclusion} concludes the paper.

\begin{figure*}[t]

  \centering
  \centerline{\includegraphics[scale=0.35]{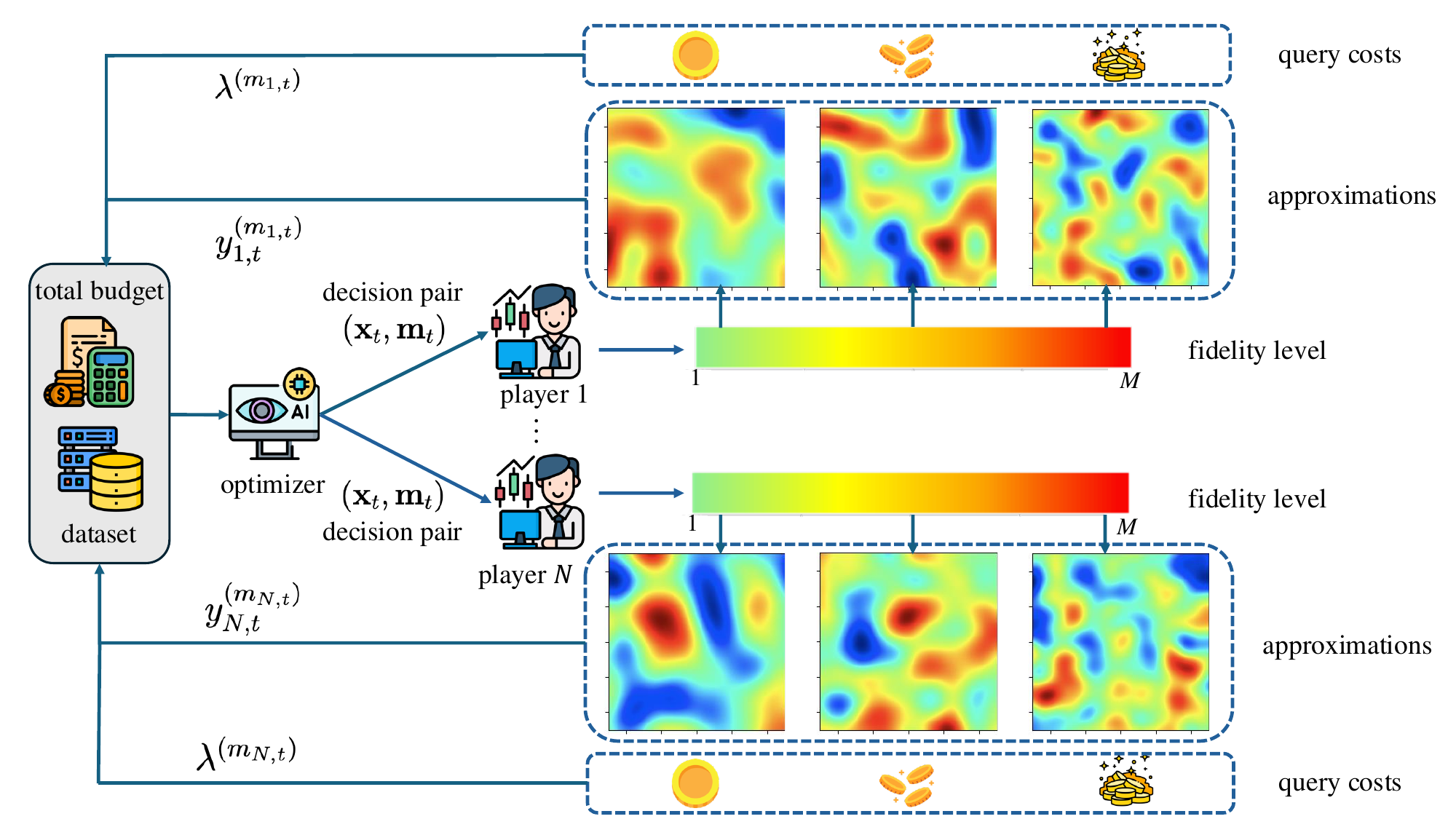}}
  \vspace{-0.2cm}
  \caption{This paper studies a setting in which a central optimizer aims at identifying an approximate pure Nash equilibrium (PNE) for a multi-player strategic game, while having access only to multi-fidelity estimators of expensive-to-evaluate black-box utility functions for the $N$ players. At any time $t$, the central optimizer selects an action profile $\mathbf{x}_{t}$ and a fidelity vector $\mathbf{m}_{t}=[m_{1,t},...,m_{N,t}]$. As a result, the optimizer receives noisy utility feedback $y_{n,t}^{(m_{n,t})}$ about the corresponding utility value $u^{(m_{n,t})}_n(\mathbf{x}_{t})$ at fidelity level $m_{n,t}$ incurring a query cost $\lambda^{(m_{n,t})}$, for all players $n\in\mathcal{N}$. The goal is to approach a solution in the $\epsilon^*$-PNE set \eqref{eq: epsilon pne}, where $\epsilon^*$ is the smallest achievable dissatisfaction level, while abiding by a total query budget $\Lambda$.}
  \label{fig: intro flow}
\vspace{-0.3cm}
\end{figure*}

\section{Problem Statement}\label{sec: pf}

\subsection{Setting}\label{ssec: sequential pne}
We consider the setting in Fig. \ref{fig: intro flow} in which a central agent is tasked with the selection of an action profile $\mathbf{x}=[\mathbf{x}_1,...,\mathbf{x}_N]$ for $N$ players, with $\mathbf{x}_n\in\mathcal{X}_n\subset\mathbbm{R}^{d_n}$ being the action for player $n\in\mathcal{N}=\{1,2,...,N\}$. The players are self-interested, and the payoff for player $n$ at action profile $\mathbf{x}\in\mathcal{X}=\prod_{n\in\mathcal{N}}\mathcal{X}_n$ is described by an unknown \emph{black-box utility function} $u_n(\mathbf{x}):\mathcal{X}\subset\mathbbm{R}^d\rightarrow\mathbbm{R}$. This can be expressed as
\begin{align}
    u_n(\mathbf{x})=u_n(\mathbf{x}_n,\mathbf{x}_{-n}),\label{eq: utility i}
\end{align}
where $\mathbf{x}_{-n}$ represents the actions performed by all players except for player $n$. 

As explained next, in order to ensure the \emph{stability} of the action profile $\mathbf{x}$, the central agent aims at selecting a \emph{pure Nash equilibrium} (PNE) of the strategic game defined by the utilities $\{u_n(\mathbf{x})\}_{n\in\mathcal{N}}$. This way, the operation of the central controller can be formulated as a form of optimization with equilibrium constraints \cite{pang2013joint,liang2015waveform,yu2017distributed}, where the equilibrium requirement pertains to the game played among the $N$ players based on the utilities in \eqref{eq: utility i}.

Under an action profile $\mathbf{x}$, each player $n\in \mathcal{N}$ has an incentive to deviate from its assigned action $\mathbf{x}_n$ if doing so can increase its utility. The strength of this incentive can be quantified by the difference
\begin{align}
    f_n(\mathbf{x})=\max\limits_{\mathbf{x}_n'\in\mathcal{X}_n}u_n(\mathbf{x}_n',\mathbf{x}_{-n})-u_n(\mathbf{x})\geq 0\label{eq: player i goal}
\end{align}
between the utility corresponding to action profile $\mathbf{x}$ and the maximum utility that player $n$ could obtain given the actions $\mathbf{x}_{-n}$ of the other players. We will refer to \eqref{eq: player i goal} as the \emph{dissatisfaction} of player $n$ at the action profile $\mathbf{x}$.

The set of PNEs for the strategic game at hand contains all action profiles $\mathbf{x}^*$ such that the dissatisfaction equals zero for each player $n\in \mathcal{N}$, i.e.,
\begin{align}
    \mathcal{X}^*:=\{\mathbf{x}^*\in\mathcal{X}|f_n(\mathbf{x}^*)=0\,\, \text{for}\,\, n\in \mathcal{N}\}.\label{eq: pure ne}
\end{align}
The PNE action profiles $\mathbf{x}^*$ in set \eqref{eq: pure ne} are incentive-compatible, such that each player has no reason to deviate from a recommended PNE action profile since following the recommendation is in their own best interest.  In practice, the PNE set may be empty, i.e., $\mathcal{X}^*=\emptyset$. Thus, it is useful to define the $\epsilon$-PNE set as
\begin{align}
    \mathcal{X}^{(\epsilon)}:=\{\mathbf{x}^{(\epsilon)}\in\mathcal{X}|f_n(\mathbf{x}^{(\epsilon)})\leq\epsilon\,\, \text{for}\,\, n\in \mathcal{N}\}\label{eq: epsilon pne}
\end{align}
for $\epsilon\geq 0$. The $\epsilon$-PNE set \eqref{eq: epsilon pne} includes solutions $\mathbf{x}^{(\epsilon)}$ from which unilateral deviations do not improve the local utility of any player by more than $\epsilon$. Thus, for an action profile $\mathbf{x}^{(\epsilon)}\in\mathcal{X}^{(\epsilon)}$, players have no incentive to deviate from the assigned action profile $\mathbf{x}^{(\epsilon)}\in\mathcal{X}^{(\epsilon)}$ as long as they consider a dissatisfaction level $\epsilon$ acceptable. The smallest achievable dissatisfaction level across all players is defined as
\begin{align}
    \epsilon^*=\inf\{\epsilon\in\mathbbm{R}|\mathcal{X}^{(\epsilon)}\neq \emptyset\}.\label{eq: epsilon definition}
\end{align}

The goal of this work is to design and analyze optimization methods that aims at approaching any $\epsilon^*$-PNE solution $\mathbf{x}^{*}\in\mathcal{X}^{(\epsilon^*)}$. As described next, this task is made challenging by the \emph{black-box} nature of the unknown utility functions. 

\subsection{Optimization}\label{ssec: optimization}
The centralized optimizer produces a series of action profiles $\mathbf{x}_t=[\mathbf{x}_{1,t},...,\mathbf{x}_{N,t}]$ at each time $t=1,...,T$ for all $N$ players. For each choice $\mathbf{x}_t$, the optimizer receives a noisy estimate of the utilities $\{u_n(\mathbf{x}_t)\}_{n=1}^N$. Specifically, we adopt a \emph{multi-fidelity} model, whereby the utilities can be measured at one of $M$ fidelity levels $\mathcal{M}=\{1,...,M\}$. Level $m=1$ corresponds to the \emph{lowest-fidelity} measurement, while $m=M$ to the \emph{highest-fidelity} level. Different fidelity levels may correspond to simulators with varying levels of accuracy for physical experiments, or to varying precision levels of measuring instruments. Accordingly, the central optimizer selects a fidelity level vector $\mathbf{m}_t=[m_{1,t},...,m_{N,t}]$ at each time $t$, along with the action profile $\mathbf{x}_t$, to form a decision pair $(\mathbf{x}_t,\mathbf{m}_t)$.

Define as $u_n^{(m)}(\mathbf{x})$ the approximation at fidelity level $m$ of the utility of player $n$. We have 
\begin{align}
    u_n^{(M)}(\mathbf{x})=u_n(\mathbf{x})\label{eq: true func}
\end{align}
for all $n\in\mathcal{N}$, so that the maximum-fidelity version coincides with the function itself. When the optimizer selects a decision pair $(\mathbf{x}_t,\mathbf{m}_{t})$ with fidelity $m_{n,t}$ for player $n$, it receives the noisy feedback
\begin{align}
    y_{n,t}^{(m_{n,t})}=u^{(m_{n,t})}_n(\mathbf{x}_{t})+z_{n,t},\label{eq: noisy utility}
\end{align}
where the observation noise variable $z_{n,t}\sim\mathcal{N}(0,\sigma^2)$ is independent across different players. As a result, at the end of time $t$, the optimizer updates the dataset of observations for each player $n\in\mathcal{N}$ as
\begin{align}
    \mathcal{D}_{n,t}=\Big\{(\mathbf{x}_{1},m_{n,1},y_{n,1}^{(m_{n,1})}),...,(\mathbf{x}_{t},m_{n,t},y_{n,t}^{(m_{n,t})})\Big\}.\label{eq: local dataset}
\end{align}

\begin{figure}[t]

  \centering
  \centerline{\includegraphics[scale=0.25]{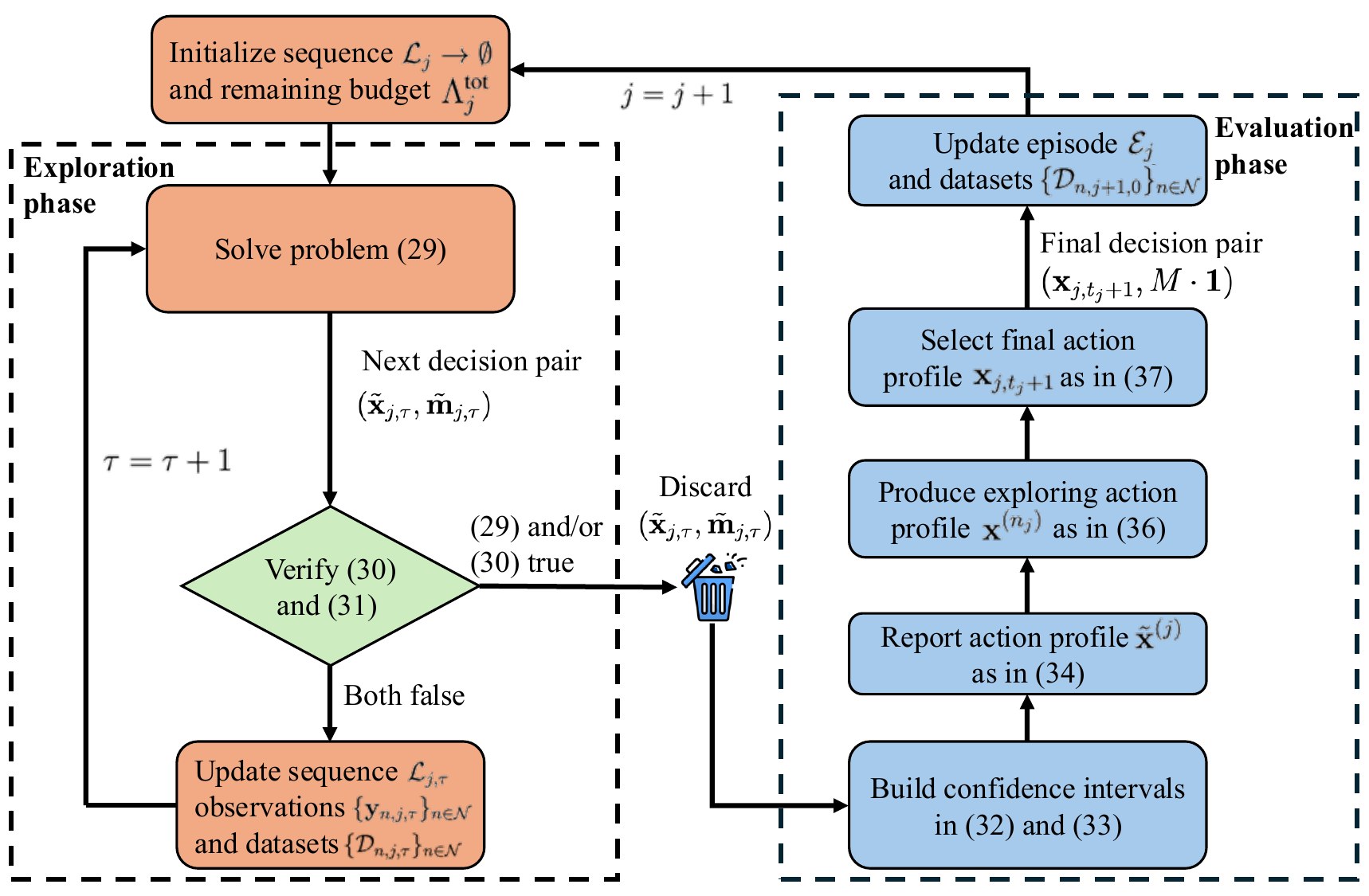}}
  \vspace{-0.2cm}
  \caption{Block diagram illustrating the operation of the proposed MF-UCB-PNE policy MF-UCB-PNE operates via a sequence of  episodes $\mathcal{E}_j$ consisting of an exploration phase and of an evaluation phase.}
  \label{fig: block diagram}
\vspace{-0.3cm}
\end{figure}

The \emph{query cost} of evaluating the utility function $u_n^{(m)}(\mathbf{x})$ at fidelity level $m$ is defined as $\lambda^{(m)}>0$ for all players $n\in\mathcal{N}$. The query costs are ordered from lowest fidelity to highest fidelity in an increasing order, i.e.,
\begin{align}
    \lambda^{(1)}\leq\lambda^{(2)}\leq ...\leq\lambda^{(M)}=1, \label{eq: cost list}
\end{align}
where we have set the maximum cost $\lambda^{(M)}=1$ without loss of generality. The total query cost is constrained by a pre-determined total query \emph{budget} $\Lambda$ as
\begin{align}
    \sum_{n=1}^N\sum_{t=1}^T\lambda^{(m_{n,t})}\leq\Lambda.\label{eq: budget}
\end{align}

Given the pre-determined query budget $\Lambda$, our goal is to design a selection strategy $\pi$ mapping the collected observations $\mathcal{D}_{t-1}=\cup_{n\in\mathcal{N}}\mathcal{D}_{n,t-1}$ to the decision pair $(\mathbf{x}_t,\mathbf{m}_t)$ at all time steps $t=1,...,T$, such that the optimizer approaches an $\epsilon^*$-PNE solution in the set \eqref{eq: epsilon pne}. 

To this end, as in the \emph{single-objective} framework studied in \cite{song2019general}, we define a proxy optimization objective in terms of a reward function that accounts for the dissatisfaction level associated with the solutions evaluated by the optimizer at the maximum fidelity, as well as for the cost of exploration at lower fidelity levels. To proceed, fix a constant $C$ to be any upper bound on the maximum value of the dissatisfaction: 
\begin{align}
        C\geq \max\limits_{n\in\mathcal{N},\mathbf{x}\in\mathcal{X}}f_n(\mathbf{x}).\label{eq: constant C}
    \end{align}
The reward is then defined by distinguishing between \emph{exploration} and \emph{evaluation} rounds: 
\begin{itemize}
    \item \emph{Evaluation round}: In an evaluation round $t$, the fidelity levels $\mathbf{m}_t$ are set at the maximum level $M$ for all players, i.e., $\mathbf{m}_t=M\cdot\mathbf{1}$, where $\mathbf{1}$ is the all-one vector. This allows the optimizer to estimate the actual utilities of all players. At any evaluation round $t$, the accrued reward is the normalized dissatisfaction improvement 
    \begin{align}
        r(\mathbf{x}_t)=\Big[C-\max_{n\in\mathcal{N}}f_n(\mathbf{x}_t)\Big]/C\label{eq: normalized reward}
    \end{align}associated with the selected action profile $\mathbf{x}_t$. By \eqref{eq: constant C}, the reward \eqref{eq: normalized reward} is non-negative. In particular, if \eqref{eq: normalized reward} holds with equality, it ranges from $0$ (maximum dissatisfaction) to $1-\epsilon^*/C$ ($\epsilon^*$-PNE).
    \item \emph{Exploration round}: In an exploration round $t$, there exists at least one fidelity level $m_n$ smaller than $M$ for some player $n\in\mathcal{N}$. In this case, following \cite{song2019general}, we set the reward to $0$. This way, observations at lower fidelity levels are leveraged purely for exploration, with the goal of enhancing future rewards \eqref{eq: normalized reward} to be obtained in subsequent evaluation rounds.

\end{itemize}

Overall, given a selection $(\mathbf{x},\mathbf{m})$, the \emph{reward} accrued by the optimizer is defined as
\begin{align}
    r(\mathbf{x},\mathbf{m})=\begin{cases}\frac{C-\max_{n\in\mathcal{N}}f_n(\mathbf{x})}{C} & \text{if}\,\, \mathbf{m}=M\cdot\mathbf{1}\quad\text{(evaluation)} \\ 0 & \text{otherwise}\quad\text{(exploration).}\label{eq: additive mf}\end{cases}
\end{align}
With this proxy objective, the design of the policy $\pi$ targets the optimization problem 
\begin{subequations}\label{eq: loss goal}
\begin{equation}
    \max\limits_{\pi,T}\sum_{t=1}^Tr(\pi(\mathcal{D}_{t-1}))\label{eq: min cumulative loss}
\end{equation}
\begin{equation}
    \text{s.t.}\quad \sum_{n=1}^N\sum_{t=1}^T\lambda^{(m_{n,t})}\leq\Lambda\label{eq: budget constraint}
\end{equation}
\end{subequations}under the budget cost \eqref{eq: budget}, where the policy $\pi$ return the actions
\begin{align}
    (\mathbf{x}_t,\mathbf{m}_t)=\pi(\mathcal{D}_{t-1})
\end{align}
at each time $t=1,...,T.$

An idealized central optimizer would know at least one action profile $\mathbf{x}^*$ in the target $\epsilon^*$-PNE set $\mathcal{X}^{(\epsilon^*)}$. Such an optimizer could hence avoid exploration steps, and choose to evaluate directly the action profile $\mathbf{x}^*\in\mathcal{X}^{(\epsilon^*)}$ for all time steps allowed by the budget, obtaining a reward value \eqref{eq: min cumulative loss} equal to $\Lambda(1-\epsilon^*/C)/N$. Note, in fact, that the maximum number of evaluation steps is $\Lambda/N$, since the cost of each evaluation step equals $N$. Therefore, the quality of a policy $\pi$ designed to target problem \eqref{eq: loss goal}  can be measured by the \emph{regret}
\begin{align}
    R(\Lambda)=\frac{\Lambda}{N}\Big(1-\frac{\epsilon^*}{C}\Big)-\sum_{t=1}^Tr(\pi(\mathcal{D}_{t-1})).\label{eq: regret goal}
\end{align}

\section{Multi-Fidelity BO for Strategic Games}\label{sec: mfbo for pne}
In this paper, we propose a multi-fidelity policy $\pi$, referred to as MF-UCB-PNE, that builds on a multiple-output Gaussian process (MOGP) surrogate model for the utility functions $\{u_n(\cdot)\}_{n\in\mathcal{N}}$ \cite{liu2018remarks}. MF-UCB-PNE operates in episodes, with each episode consisting of an exploration phase and a final evaluation time step \cite{song2019general}. This section introduces MF-UCB-PNE by first presenting the adopted surrogate MOGP model, and then detailing the exploration and evaluation phases. We start with an overview of the MF-UCB-PNE policy.

\subsection{Overview of MF-UCB-PNE}\label{ssec: mf-ucb-pne}
As shown in Fig. \ref{fig: block diagram}, MF-UCB-PNE balances the exploration of low fidelity observations to collect information about the utility functions with the evaluation of utilities at the maximum fidelity level to identify well-performing action profiles approaching the $\epsilon^*$-PNE set \eqref{eq: epsilon pne}. To this end, the operation of the MF-UCB-PNE policy is organized in episodes encompassing an exploration phase and a single evaluation step.

\begin{definition}[Episode]\label{definition: episode}
    An episode $\mathcal{E}_j$ consists of a sequence of decision pairs $\mathcal{E}_j=\{(\mathbf{x}_{j,1},\mathbf{m}_{j,1}),\\...,(\mathbf{x}_{j,t},\mathbf{m}_{j,t}),(\mathbf{x}_{j,t+1},\mathbf{m}_{j,t+1})\}$ divided into exploration and evaluation phases:
    \begin{itemize}
        \item \textit{Exploration phase}: The first $t$ decisions 
        \begin{align}
            \mathcal{L}_j=\{(\mathbf{x}_{j,\tau},\mathbf{m}_{j,\tau})\}_{\tau=1}^t\label{eq: L j}
        \end{align}
        evaluate at least one utility at fidelity level lower than the maximum level $M$, i.e., $\mathbf{m}_{j,\tau}\neq M\cdot\mathbf{1}$, forming the exploration phase. 
        \item Evaluation phase: The final decision pair $(\mathbf{x}_{j,t+1},\mathbf{m}_{j,t+1})$ corresponds to the evaluation phase, as it evaluates all the utility functions at the maximum fidelity level, i.e., $\mathbf{m}_{j,t+1}=M\cdot\mathbf{1}$.
    \end{itemize}
    
\end{definition}

Following the Definition \ref{definition: episode}, MF-UCB-PNE carries out a sequence of episodes $\{\mathcal{E}_j\}$, with each episode $\mathcal{E}_j$ consisting of a sequence of decision pairs $\mathcal{L}_{j}=\{(\mathbf{x}_{j,1},\mathbf{m}_{j,1}),...,(\mathbf{x}_{j,t_j},\mathbf{m}_{j,t_j})\}$ exploring lower fidelity utilities and of a final decision pair $(\mathbf{x}_{j,t_j+1},M\cdot\mathbf{1})$ selecting maximum fidelity for all players. Accordingly, the accumulated query cost in episode $\mathcal{E}_j$ is obtained as
\begin{align}
    \Lambda_j=\Lambda_{\mathcal{L}_{j}}+N,
\end{align}with $\Lambda_{\mathcal{L}_{j}}$ being the cumulative query cost incurred on the sequence $\mathcal{L}_{j}$ for the exploration phase
\begin{align}
    \Lambda_{\mathcal{L}_{j}}=\sum_{n=1}^N\sum_{\tau=1}^{t_j}\lambda^{(m_{n,j,\tau})},\label{eq: L j budget}
\end{align}
and $N$ being the cost of the evaluation phase.

\begin{algorithm}[t!]
\caption{MF-UCB-PNE}\label{table: mf-ucb-pne}
\SetKwInOut{Input}{Input}
\Input{Query costs $\{\lambda^{(m)}\}_{m=1}^M$, query budget $\Lambda$}
\SetKwInOut{Output}{Output}
\Output{Optimized solution $\mathbf{x}^*$}\
Initialize episode $j=1$, observation dataset $\mathcal{D}_{n,j,0}=\emptyset$, and remaining budget $\Lambda^{\text{tot}}_j=\Lambda$\\
\While{\emph{$\Lambda^{\text{tot}}_j\geq N$}}{
Run MF-UCB-PNE-Episode$(\Lambda^{\text{tot}}_j,\{\mathcal{D}_{n,j,0}\}_{n\in\mathcal{N}})$ (Algorithm \ref{table: lf-mi-max}) to obtain the sequence of decision pairs $\mathcal{L}_{j}$, observations $\{\mathbf{y}_{n,j,t_j}\}_{n\in\mathcal{N}}$, accumulated cost $\Lambda_{\mathcal{L}_j}$ and updated dataset $\{\mathcal{D}_{n,j,t_j}\}_{n\in\mathcal{N}}$\\
Obtain the utility confidence interval $\{\mathcal{U}_{n,j}\}_{n\in\mathcal{N}}$ as in \eqref{eq: u j confidence interval}\\
Obtain the dissatisfaction confidence interval $\{\mathcal{C}_{n,j}\}_{n\in\mathcal{N}}$ as in \eqref{eq: j Delta ci} \\
Report the action profile $\tilde{\mathbf{x}}^{(j)}$ via \eqref{eq: j reported x}\\
Select the exploring action profile $\mathbf{x}^{(n_j)}$ via \eqref{eq: exploring strategy profile}\\
Produce the final decision pair $(\mathbf{x}_{j,t_j+1},M\cdot\mathbf{1})$ for current episode $\mathcal{E}_j$ with action profile $\mathbf{x}_{j,t_j+1}$ selected via \eqref{eq: x t+1}\\
Query $(\mathbf{x}_{j,t_j+1},M\cdot\mathbf{1})$ and update datasets as $\mathcal{D}_{n,j+1,0}=\mathcal{D}_{n,j,t_j}\cup(\mathbf{x}_{j,t_j+1},M,u_n(\mathbf{x}_{j,t_j+1})+z_{n,j,t_j+1})$ for all $n\in\mathcal{N}$\\
Update the MOGP posteriors using $\mathcal{D}_{n,j,t_j+1}$ as in \eqref{eq: mogp mean} and \eqref{eq: mogp variance} for all $n\in\mathcal{N}$\\
Calculate the remaining budget $\Lambda^{\text{tot}}_{j+1}=\Lambda^{\text{tot}}_j-\Lambda_{\mathcal{L}_{j}}-N$\\
Set episode index $j=j+1$
}
Return $\mathbf{x}^{*}=\mathbf{x}_{T_j}$\\
\end{algorithm}

\subsection{Multiple-Output Gaussian Process}\label{ssec: mogp}
For both exploration and evaluation phase, MF-UCB-PNE builds on probabilistic surrogate models for the multi-fidelity utility functions $\{\mathbf{u}_n(\mathbf{x})=[u_n^{(1)}(\mathbf{x}),...,u_n^{(M)}(\mathbf{x})]^{\sf T}\}$ for all players $n\in\mathcal{N}$. Under this model, the vectors $\mathbf{u}_n(\mathbf{x})$ are assumed to be independent across the player index $n$; while the approximations $\{u_n^{(m)}(\mathbf{x})\}_{m\in\mathcal{M}}$ are allowed to be generally correlated across the fidelity index $m$ for each player $n$, following a zero-mean MOGP with kernel function $k((\mathbf{x},m),(\mathbf{x},m'))$ \cite{kennedy2000predicting}. A common MOGP follows the auto-regressive model \cite{kennedy2000predicting,raissi2016deep},
\begin{align}
    u_n^{(m)}(\mathbf{x})= \rho^{(m)} u_n^{(m+1)}(\mathbf{x})+\sqrt{1-(\rho^{(m)})^2}q^{(m)}_n(\mathbf{x}),\label{eq: ar model}
\end{align}
where $0<\rho^{(m)}<1$ is a correlation coefficient, and we have $u_n^{(M)}(\mathbf{x})=u_n(\mathbf{x})\sim\mathcal{GP}(0,\kappa(\mathbf{x},\mathbf{x}'))$, and $q^{(m)}_n(\mathbf{x})\sim\mathcal{GP}(0,\kappa^{(m)}(\mathbf{x},\mathbf{x}'))$, where the kernel functions are often instantiated as the radial basis function (RBF) kernels
\begin{align}
    &\kappa(\mathbf{x},\mathbf{x}')=\exp(-h||\mathbf{x}-\mathbf{x}'||^2),\label{eq: rbf kernels}\\
    &\kappa^{(m)}(\mathbf{x},\mathbf{x}')=\exp(-\zeta^{(m)}||\mathbf{x}-\mathbf{x}'||^2)\label{eq: fidelity kernel}
\end{align}
with lengthscales $h>0$ and $\zeta^{(m)}>0$, respectively. 

Let $\mathbf{y}_{n,t}=[y_{n,1}^{(m_{n,1})},...,y_{n,t}^{(m_{n,t})}]^{\sf T}$ denote the $t\times 1$ observation vector obtained based on the selected actions profiles $\mathbf{X}_t=[\mathbf{x}_1,...,\mathbf{x}_t]$ and fidelity levels $\mathbf{M}_{n,t}=[m_{n,1},...,m_{n,t}]$ for each player $n$ up to time $t$. Using the $t\times t$ covariance matrix $\mathbf{K}(\mathbf{X}_t,\mathbf{M}_{n,t})$ with the $(i,j)$-th element $k((\mathbf{x}_i,m_{n,i}),(\mathbf{x}_j,m_{n,j}))$, the MOGP posterior distribution of the approximation $u_n^{(m)}(\mathbf{x})$ at any action profile $\mathbf{x}$ and fidelity level $m$ given the collected observations \eqref{eq: local dataset} is Gaussian with mean
\begin{equation}
    \mu^{(m)}_{n,t}(\mathbf{x})=\mathbf{k}(\mathbf{x},m)^{\sf T}(\tilde{\mathbf{K}}(\mathbf{X}_{t},\mathbf{M}_{n,t}))^{-1}\mathbf{y}_{n,t},\label{eq: mogp mean}
\end{equation}
and variance
\begin{align}
    [\sigma_{n,t}^{(m)}(\mathbf{x})]^2=&k((\mathbf{x},m),(\mathbf{x},m))-\nonumber\\&\mathbf{k}(\mathbf{x},m)^{\sf T}(\tilde{\mathbf{K}}(\mathbf{X}_{t},\mathbf{M}_{n,t}))^{-1}\mathbf{k}(\mathbf{x},m),\label{eq: mogp variance}
\end{align}
where the $t\times 1$ cross-variance vector $\mathbf{k}(\mathbf{x},m)=[k((\mathbf{x},m),(\mathbf{x}_{1},m_{n,1})),...,k((\mathbf{x},m),(\mathbf{x}_{t},m_{n,t}))]^{\sf T}$ and we defined $\tilde{\mathbf{K}}(\mathbf{X}_{t},\mathbf{M}_{n,t})=\mathbf{K}(\mathbf{X}_{t},\mathbf{M}_{n,t})+\sigma^2\mathbf{I}$. 

Overall, the MOGP mean $\mu^{(m)}_{n,t}(\mathbf{x})$ in \eqref{eq: mogp mean} provides an estimated value for the approximation $u_n^{(m)}(\mathbf{x})$ based on the collected dataset \eqref{eq: local dataset}, while the variance $[\sigma_{n,t}^{(m)}(\mathbf{x})]^2$ quantifies the corresponding uncertainty of the estimate.

\subsection{Exploration Phase}\label{ssec: exploration phase}
With this background in place, we can now define the two phases implemented by MF-UCB-PNE in each episode $\mathcal{E}_j$ (see Fig. \ref{fig: block diagram}). In the exploration phase, we adopt an \emph{information-theoretic} criterion for selecting the sequence of actions. To elaborate, denote as $\mathcal{D}_{n,j,\tau-1}$ the data collected up to time $\tau-1$ in episode $j$. By defined next, at the beginning of time $\tau\geq 1$ in episode $\mathcal{E}_j$, the optimizer's policy $\pi$ selects the next pair $(\mathbf{x}_{j,\tau},\mathbf{m}_{j,\tau})$ to maximize the ratio of the overall information obtained about the true value of the utility functions $\{u_n(\mathbf{x})\}_{n\in\mathcal{N}}$ and the cost of the selection $(\mathbf{x}_{j,\tau},\mathbf{m}_{j,\tau})$.

The information about the utility function is quantified by the \emph{mutual information} between the true utility value $u_n(\mathbf{x}_{j,\tau})$ and the corresponding observation $y_n^{(m_{n,j,\tau})}$ given the available dataset $\mathcal{D}_{n,j,\tau-1}$. This is defined as \cite{srinivas2012information}
\begin{align}
    &\mathrm{I}(y_n^{(m_{n,j,\tau})};u_n(\mathbf{x}_{j,\tau})|\mathbf{x}_{j,\tau},m_{n,j,\tau},\mathcal{D}_{n,j,\tau-1})\nonumber\\&=\frac{1}{2}\ln\Bigg(1+\frac{[\sigma_{n,j,\tau-1}^{(m_{n,j,\tau})}(\mathbf{x}_{j,\tau})]^2}{\sigma^2}\Bigg),\label{eq: mutual information}
\end{align}
where the GP posterior variance $[\sigma_{n,j,\tau-1}^{(m_{n,j,\tau})}(\mathbf{x}_{j,\tau})]^2$ is given in \eqref{eq: mogp variance}. The mutual information \eqref{eq: mutual information} measures the uncertainty reduction about utility $u_n(\mathbf{x}_{j,\tau})$ that can be obtained by choosing the pair $(\mathbf{x}_{j,\tau},\mathbf{m}_{j,\tau})$ at time step $\tau$ in episode $j$.

Define as 
\begin{align}
    \Lambda^{\text{tot}}_j=\Lambda-\sum_{j'=1}^{j-1}\Lambda_{j'}\label{eq: remaining budget}
\end{align}
the remaining query budget before starting episode $\mathcal{E}_j$, and as
\begin{align}
    \Lambda^{\text{tot}}_{j,\tau}=\Lambda^{\text{tot}}_j-\sum_{n=1}^N\sum_{\tau'=1}^{\tau-1}\lambda^{(m_{n,j,\tau'})}\label{eq: remaining budget at tau}
\end{align}
the remaining budget at step $\tau$ of episode $\mathcal{E}_j$. The following events may occur.
\begin{enumerate}
    \item \emph{Query budget insufficient for exploration}: If the budget is insufficient to accommodate another exploration step followed by the final evaluation step, i.e., if
   \begin{align}
        \Lambda^{\text{tot}}_j<N(\lambda^{(1)}+1),\label{eq: insufficient budget}
    \end{align}the exploration phase is concluded, setting $t_j=\tau-1$.
    \item \emph{Query budget sufficient for exploration}: If \eqref{eq: insufficient budget} does not hold, MF-UCB-PNE obtains a \emph{candidate solution} by addressing the problem of maximizing the ratio between the sum of the mutual information \eqref{eq: mutual information} accrued on all players and the corresponding total query costs, i.e.,
    \begin{subequations}\label{eq: constrained opt}
    \begin{align}
        &(\tilde{\mathbf{x}}_{j,\tau},\tilde{\mathbf{m}}_{j,\tau})\nonumber\\&=\argmax\limits_{\mathbf{x}\in\mathcal{X},\mathbf{m}\in\mathcal{M}}\frac{\sum_{n=1}^N \mathrm{I}(y_n^{(m_n)};u_n(\mathbf{x})|\mathbf{x},m_n,\mathcal{D}_{n,j,\tau-1})}{\sum_{n=1}^N\lambda^{(m_n)}}\label{eq: select exploring decision}
    \end{align}
    \begin{equation}
        \text{s.t.}\quad \sum_{n=1}^N\lambda^{(m_n)}\leq\Lambda^{\text{tot}}_{j,\tau}-N.\label{eq: exploring condition 1}
    \end{equation}
    \end{subequations}
    The constraint \eqref{eq: exploring condition 1} ensures that the remaining query budget after the current query is sufficient for at least an evaluation step at the maximum fidelity level, which has cost $N$. Having solved problem \eqref{eq: constrained opt}, the following options may be realized.
    \begin{enumerate}
        \item \emph{Large fraction of maximum-fidelity evaluations}: If a sufficiently large fraction of fidelity levels in vector $\mathbf{m}_{j,\tau}$ are equal to $M$, i.e., if the inequality
        \begin{align}
            \frac{\sum_{n=1}^N\mathbbm{1}(m_{n,j,\tau}=M)}{N}\geq\eta\label{eq: exploring condition 2}
        \end{align}
        holds for some threshold $\eta\in[1/N,1]$, then MF-UCB-PNE sets $t_j=\tau-1$ and enters the evaluation phase.
        \item \emph{Small accumulated information gain-cost ratio}: If the information gain-cost ratio accumulated on selected sequence $\mathcal{L}_{j,\tau}=\{(\mathbf{x}_{j,1},\mathbf{m}_{j,1}),...,\\(\mathbf{x}_{j,\tau},\mathbf{m}_{j,\tau})\}$ in  episode $j$ satisfies the inequality
        \begin{align}
        &\frac{\sum_{n=1}^N\mathrm{I}(\mathbf{y}_{n,j,\tau};\{u_n(\mathbf{x})\}_{\mathbf{x}\in\mathcal{L}_{j,\tau}}|\mathcal{L}_{j,\tau},\mathcal{D}_{n,j,0})}{\sum_{n=1}^N\sum_{\tau'=1}^{\tau}\lambda^{(m_{n,j,\tau'})}}\nonumber\\&<1/\sqrt{\Lambda^{\text{tot}}_j},\label{eq: exploring condition 3}
    \end{align}
    then MF-UCB-PNE sets $t_j=\tau-1$ and enters the evaluation phase.
    \item \emph{All other cases}: If neither of the conditions \eqref{eq: exploring condition 2} or \eqref{eq: exploring condition 3} is satisfied, the algorithm sets $(\tilde{\mathbf{x}}_{j,\tau},\tilde{\mathbf{m}}_{j,\tau})$ via \eqref{eq: select exploring decision} and continues the exploration phase.
    \end{enumerate}
\end{enumerate}

\subsection{Evaluation Phase}\label{ssec: evaluation phase}
In the final evaluation step, MF-UCB-PNE adopts a strategy that follows UCB-PNE \cite{tay2023no} to select the candidate action profile $\mathbf{x}_{j,t_j+1}$. Accordingly, MF-UCB-PNE starts by computing a confidence interval for true utility function $u_n(\mathbf{x})$ of each player $n$ as
\begin{align}
    \mathcal{U}_{n,j}(\mathbf{x})&=[\check{u}_{n,j}(\mathbf{x}), \hat{u}_{n,j}(\mathbf{x})]\nonumber\\&=[\mu^{(M)}_{n,j,t_j}(\mathbf{x})-\beta_{n,j}\sigma^{(M)}_{n,j,t_j}(\mathbf{x}), \mu^{(M)}_{n,j,t_j}(\mathbf{x})\nonumber\\&\hspace{0.5cm}+\beta_{n,j}\sigma^{(M)}_{n,j,t_j}(\mathbf{x})],\label{eq: u j confidence interval}
\end{align}
where $\beta_{n,j}>0$ is a scaling parameter. We will elaborate on the selection of parameter $\beta_{n,j}$ in the next section. Using the confidence interval $\mathcal{U}_{n,j}(\mathbf{x})$ in \eqref{eq: u j confidence interval} for the utility function, the optimizer also calculates confidence interval for the dissatisfaction $f_n(\mathbf{x})$ in \eqref{eq: player i goal} at each player $n$ as
\begin{align}
    \mathcal{C}_{n,j}(\mathbf{x})&=[\check{f}_{n,j}(\mathbf{x}),\hat{f}_{n,j}(\mathbf{x})]\nonumber\\&=\Big[\max\limits_{\mathbf{x}_n'\in\mathcal{X}}\check{u}_{n,j}(\mathbf{x}_n',\mathbf{x}_{-n})-\hat{u}_{n,j}(\mathbf{x}), \nonumber\\&\hspace{0.6cm}\max\limits_{\mathbf{x}_n'\in\mathcal{X}}\hat{u}_{n,j}(\mathbf{x}_n',\mathbf{x}_{-n})-\check{u}_{n,j}(\mathbf{x})\Big].\label{eq: j Delta ci}
\end{align}

The optimizer then reports an action profile $\tilde{\mathbf{x}}^{(j)}$ for possible execution at the end of episode $j$ as
\begin{align}
    \tilde{\mathbf{x}}^{(j)}=\arg\min\limits_{\mathbf{x}\in\mathcal{X}}\max\limits_{n\in\mathcal{N}}\check{f}_{n,j}(\mathbf{x}).\label{eq: j reported x}
\end{align}
The reported action profile \eqref{eq: j reported x} minimizes the optimistic estimate $\check{f}_{n,j}(\mathbf{x})$ of the dissatisfaction function \eqref{eq: player i goal} for the player $n\in\mathcal{N}$ that has the strongest incentive to deviate. 

\begin{algorithm}[t!]
\caption{MF-UCB-PNE-Episode}\label{table: lf-mi-max}
\SetKwInOut{Input}{Input}
\Input{Remaining budget $\Lambda^{\text{tot}}_j$, datasets $\{\mathcal{D}_{n,j,0}\}_{n\in\mathcal{N}}$, threshold $\eta$}
\SetKwInOut{Output}{Output}
\Output{Sequence $\mathcal{L}_j$, observations $\{\mathbf{y}_{n,j,t_j}\}_{n\in\mathcal{N}}$, accumulated cost $\Lambda_{\mathcal{L}_j}$ and updated dataset $\{\mathcal{D}_{n,j,t_j}\}_{n\in\mathcal{N}}$}\
Initialize episode time $\tau=1$, sequence $\mathcal{L}_{j,0}=\emptyset$, and accumulated cost $\Lambda_{\mathcal{L}_j}=0$\\
\While{\emph{Not End}}{
Select decision pair $(\tilde{\mathbf{x}}_{j,\tau},\tilde{\mathbf{m}}_{j,\tau})$ by solving \eqref{eq: constrained opt}\\
      \If{inequality \eqref{eq: exploring condition 2} or \eqref{eq: exploring condition 3} hold true}
      {
        break 
      }
      
      \Else{
      
Update datasets $\mathcal{D}_{n,j,\tau}$ with observation $y_{n}^{(m_{n,j,\tau})}$ for all $n\in\mathcal{N}$\\
Update sequence $\mathcal{L}_{j,\tau}=\mathcal{L}_{j,\tau-1}\cup(\tilde{\mathbf{x}}_{j,\tau},\tilde{\mathbf{m}}_{j,\tau})$\\
Update the MOGP posteriors using $\mathcal{D}_{n,j,\tau}$ as in \eqref{eq: mogp mean} and \eqref{eq: mogp variance} for all $n\in\mathcal{N}$\\
Calculate the accumulated cost $\Lambda_{\mathcal{L}_j}=\Lambda_{\mathcal{L}_j}+\sum_{n=1}^N\lambda^{(m_{n,j,\tau})}$ \\
Set episode time $\tau=\tau+1$
}
Obtain $\mathcal{L}_j=\mathcal{L}_{j,\tau}$ and $\{\mathcal{D}_{n,j,t_j}=\mathcal{D}_{n,j,\tau}\}_{n\in\mathcal{N}}$} 
Return $\mathcal{L}_j$, $\{\mathcal{D}_{n,j,t_j}\}_{n\in\mathcal{N}}$, $\Lambda_{\mathcal{L}_j}$ and $\{\mathbf{y}_{n,j,t_j}\}_{n\in\mathcal{N}}$\\
\end{algorithm}

To allow for exploration at the maximum fidelity level $M$, the optimizer also searches for the player $n_{j}$ with the maximum dissatisfaction upper confidence bound under the reported action profile $\tilde{\mathbf{x}}^{(j)}$ in \eqref{eq: j reported x}, i.e.,
\begin{align}
    n_{j}=\arg\max\limits_{n\in\mathcal{N}}\hat{f}_{n,j}(\tilde{\mathbf{x}}^{(j)}),\label{eq: worst player}
\end{align}
and then updates the reported action profile $\tilde{\mathbf{x}}^{(j)}$ to the \emph{exploring action profile}
\begin{align}
    \mathbf{x}^{(n_j)}=\Big(\tilde{\mathbf{x}}^{(j)}_{-n_j}, \arg\max\limits_{\mathbf{x}'_{n_j}\in\mathcal{X}_{n_j}}\hat{u}_{n_j,j}(\mathbf{x}'_{n_{j}},\tilde{\mathbf{x}}^{(j)}_{-n_j})\Big).\label{eq: exploring strategy profile}
\end{align}
The exploring action profile $\mathbf{x}^{(n_j)}$ attempts to improve the utility gain for the player $n_j$ that the optimizer believes to have the strongest incentive to deviate from $\tilde{\mathbf{x}}^{(j)}$.

The final decision for the next action profile at the end of episode $\mathcal{E}_j$ is selected between the reported action profile $\tilde{\mathbf{x}}^{(j)}$ in \eqref{eq: j reported x} and the exploring action profile $\mathbf{x}^{(n_j)}$ in \eqref{eq: exploring strategy profile} that brings higher GP posterior variance, i.e.,
\begin{align}
    \mathbf{x}_{j,t_j+1}=\arg\max\limits_{\mathbf{x}\in\{\tilde{\mathbf{x}}^{(j)},\mathbf{x}^{(n_j)}\}}\bigg\{\max_{n\in\mathcal{N}}\Big[\sigma^{(M)}_{n,j,t_j}(\mathbf{x})\Big]^2\bigg\}.\label{eq: x t+1}
\end{align}
The rationale of the choice \eqref{eq: x t+1} refers to a double application of the principle of \emph{optimism in the face of uncertainty} \cite{lai1985asymptotically}.

Like BO, MF-UCB-PNE is a meta-algorithm that can be instantiated using different implementations for the update of the GP surrogate and for the acquisition functions \eqref{eq: constrained opt} and \eqref{eq: j reported x}. For context, in Bayesian optimization, surrogate model complexity ranges from $\mathcal{O}(t^3)$ for exact GP inference to linear complexity when using sparse GPs with inducing points. Similarly, acquisition function optimization complexity varies with the chosen approach \cite{wilson2018maximizing}. Our framework inherits this flexibility and can be implemented with the same range of computational strategies.

The key distinction from standard Bayesian optimization is that our exploitation phase requires solving for a Nash equilibrium. While this is computationally challenging (as is the optimization problem in standard BO), numerous efficient heuristic solutions exist and can be readily employed. This was comprehensively explained in our response. Finding a Nash equilibrium is, in general, computationally intractable in worst-case scenarios, even for simple two-player games. This inherent hardness means that our algorithm, like \cite{tay2023no}, can be implemented exactly in settings in which the bilevel optimization problem \eqref{eq: j reported x} can be solved via convex optimization or exhaustive search \cite{shen2023online,houska2013nonlinear,angulo2021algorithms}, while otherwise requiring the use of approximation methods \cite{sinha2021solving}.

\section{Theoretical Analysis}\label{sec: theorems}
In this section, we establish theoretical results on the performance of MF-UCB-PNE in terms of the regret \eqref{eq: regret goal}.

Previous work has addressed the analysis of the regret $R(\Lambda)$ in two special cases:
\begin{itemize}
    \item \emph{Multi-fidelity single-player case}: Reference \cite{song2019general} studied the special case of MF-UCB-PNE corresponding to $N=1$ player. In this setting, the goal is to find a maximizer for the utility function of the only player. Note that the notion of PNE does not apply in this case. The work \cite{song2019general} demonstrates that the resulting policy attains a regret that grows sublinearly with the budget $\Lambda$, ensuring the asymptotic \emph{no-regret} condition
    \begin{align}
        \lim_{\Lambda\to\infty}\frac{R(\Lambda)}{\Lambda}=0.\label{eq: no regret}
    \end{align}
    \item \emph{Single-fidelity multi-player case}: Reference \cite{tay2023no} analyzes the scenario with $M=1$, i.e., with a single fidelity level, demonstrating the validity of the limit \eqref{eq: no regret} also in this case.
\end{itemize}
In this section, we extend the proof of the sublinearity of the regret \eqref{eq: regret goal} in general case with multiple fidelities and players.

\subsection{Technical Assumptions}\label{ssec: assumptions}
We start by separating the contribution of different episodes $\{\mathcal{E}_j\}_{j=1}^J$ to the regret $R(\Lambda)$ in \eqref{eq: regret goal}. This is done by defining the episode regret $R(\mathcal{E}_j)$. As detailed next, this captures the gap between the accumulated reward \eqref{eq: additive mf} incurred within the episode $\mathcal{E}_j$ and the optimal dissatisfaction performance that can be achieved given the episode budget $\Lambda_j$, namely $(1-\epsilon^*/C)(\Lambda_j/N)$.
\begin{definition}[Episode Regret]\label{definition: episode regret}
    Denote by $\Lambda_j$ the overall query cost accumulated in episode $\mathcal{E}_j$, i.e.,
    \begin{align}
        \Lambda_j=\sum_{n\in\mathcal{N}}\sum_{\tau=1}^{t_j}\lambda^{(m_{n,j,\tau})}+N,\label{eq: Lambda j}
    \end{align}where $t_j$ is the number of exploration steps in episode $\mathcal{E}_j$. Given a total budget $\Lambda$ and a total number of episodes $J$ satisfying the inequality $\sum_{j=1}^J\Lambda_j\leq\Lambda$, the regret \eqref{eq: regret goal} of MF-UCB-PNE can be expressed as a sum over episodes as
    \begin{align}
        R(\Lambda)=\sum_{j=1}^J R(\mathcal{E}_j).\label{eq: cumulative episode regret}
    \end{align}
    Specifically, in \eqref{eq: cumulative episode regret}, the regret of the $j$-th episode, $\mathcal{E}_j$, is defined as
    \begin{align}
        R(\mathcal{E}_j)=\frac{\Lambda_j}{N}\Big(1-\frac{\epsilon^*}{C}\Big)-\Big(1-\frac{\epsilon_j}{C}\Big),\label{eq: episode regret}
    \end{align}where $\epsilon_j$ is the maximum dissatisfaction \eqref{eq: player i goal} incurred by the final action profile $\mathbf{x}_{j,t_j+1}$ across all players, i.e.,
    \begin{align}
        \epsilon_j=\max\limits_{n\in\mathcal{N}}f_n(\mathbf{x}_{j,t_j+1}).\label{eq: maximum complaint}
    \end{align}
\end{definition}


As in \cite{tay2023no}, we make the following regularity assumptions on the true utility functions of each player.
\begin{assumption}[Function space]\label{assumption: RKHS}
    The utility function $u_n(\mathbf{x})$ for each player $n\in\mathcal{N}$ is assumed to lie in the reproducing kernel Hilbert space (RKHS) $\mathcal{H}_k$ associated with the same kernel function assumed by the MOGP prior. Therefore, the utility function of each player $n\in\mathcal{N}$ can be expressed as
    \begin{align}
        u_n(\mathbf{x})=\sum_{s=1}^{S_n}a_{s}k((\mathbf{x},M),(\mathbf{x}_{s},M)),\label{eq: rkhs function}
    \end{align}
    for some real-valued coefficients $\{a_{s}\}_{s=1}^{S_n}$ and inducing vectors $\{\mathbf{x}_{s}\}_{s=1}^{S_n}$, where $S_n$ may be unbounded.
\end{assumption}

We also make the standard assumption that the complexity of each utility function $u_n(\mathbf{x})$, as measured by the RKHS norm $||u_n||_{\mathcal{H}_k}$, is bounded \cite{chowdhury2017kernelized,zhang2024bayesian}.
\begin{assumption}[RKHS norm bound]\label{assumption: rkhs norm}
    The RKHS norm of the utility function $u_n(\mathbf{x})$ satisfies the inequality $||u_n||_{\mathcal{H}_k}\leq B$, for all players $n\in\mathcal{N}$ for some known upper bound $B>0$.
\end{assumption}

\begin{figure}[t]

  \centering
  \centerline{\includegraphics[scale=0.24]{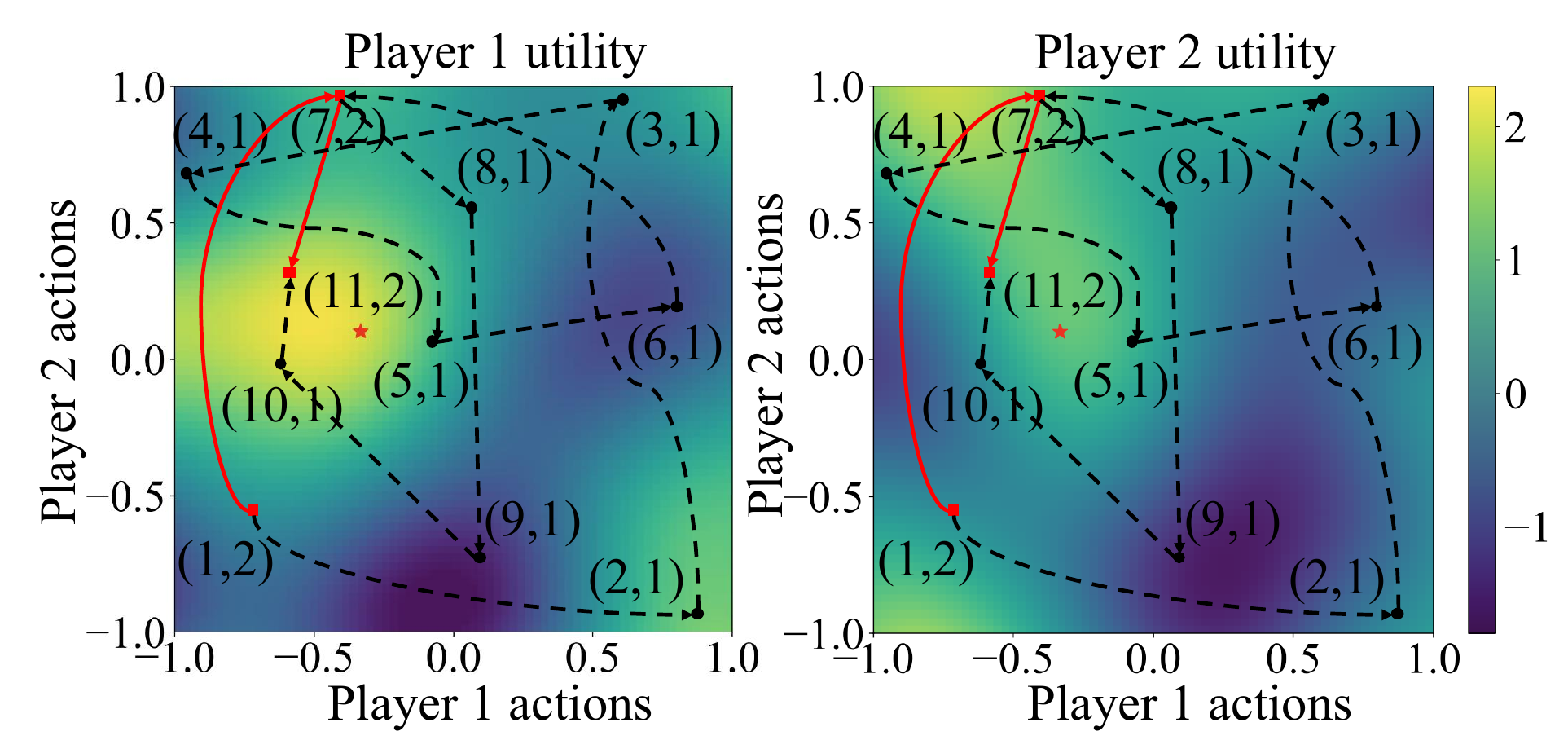}}
 \vspace{-0.3cm}
  \caption{Visualization of MF-UCB-PNE optimization trajectories against a heatmap of each player's utility function. The query budget is set to $\Lambda=32$, and the threshold in \eqref{eq: exploring condition 2} is set to $\eta=0.5$. The indexes $(t,m_{n,t})$ associated to each point in the heat map describe the time step $t$ increasing the counter across episodes and the fidelity decision $m_{n,t}$ assigned to player $n$ at time step $t$. The $\epsilon^*$-PNE is marked as a red star. The initial action at time $t=1$ and the actions made in evaluation phases are shown as red squares, while decisions made in the exploration phases are marked as black circles. Dashed black arrows follow the updates across time steps, while solid red arrows connect the actions obtained at the evaluation phases at the end of each episode.}
  \label{fig: trajectory}
\vspace{-0.3cm}
\end{figure}

\subsection{Regret Analysis}\label{ssec: regret analysis}
To obtain guarantees on the regret $R(\Lambda)$, we select the scaling parameter $\beta_{n,j}>0$ used in the confidence intervals \eqref{eq: u j confidence interval} and \eqref{eq: j Delta ci} as \cite{chowdhury2017kernelized}
\begin{align}
    \beta_{n,j}=B+4\sigma\sqrt{1+\gamma_{n,j}+\ln(1/\delta)},\label{eq: beta}
\end{align} 
where $\delta\in(0,1)$ is a parameter, and $\gamma_{n,j}$ is the \emph{maximal information gain} for player $n\in\mathcal{N}$ at episode $j$. The maximal information gain quantifies the maximum reduction in uncertainty about the utility function $u_n(\mathbf{x})$ that can be attained in episode $\mathcal{E}_j$, which contains $t_j+1$ timesteps. This is defined as the mutual information \cite{srinivas2012information,berkenkamp2023bayesian}
\begin{align}
    \gamma_{n,j}=\max\limits_{\mathbf{X}_{j},\mathbf{M}_{n,j}}\mathrm{I}(\mathbf{Y}_{n,j};\{u_n(\mathbf{x})\}_{\mathbf{x}\in\mathbf{X}_{j}}|\mathbf{X}_{j},\mathbf{M}_{n,j})\label{eq: Tj maximal mi}
\end{align}
between the observations $\mathbf{Y}_{n,j}=[y^{(m_{n,1,1})}_{n},...,y_n^{(m_{n,1,t_1+1})},\\...,y_n^{(m_{n,j,1})},...,y_n^{(m_{n,j,t_j+1})}]$ and the true utility values $\{u_n(\mathbf{x})\}_{\mathbf{x}\in\mathbf{X}_{j}}$ for a given selection of action profiles $\mathbf{X}_{j}=[\mathbf{x}_{1,1},...,\mathbf{x}_{1,t_1+1},...,\mathbf{x}_{j,1},...,\mathbf{x}_{j,t_j+1}]$ and fidelity levels $\mathbf{M}_{n,j}=[m_{n,1,1},...,m_{n,1,t_1+1},...,m_{n,j,1},...,m_{n,j,t_j+1}]$. Evaluating the information gain \eqref{eq: Tj maximal mi} is often computationally infeasible due to the need to optimize over action profiles $\mathbf{X}_j$ and fidelity levels $\mathbf{M}_{n,j}$. Therefore, it is typically approximated via a greedy algorithm \cite{srinivas2012information}. Using the notion $\mathcal{J}=\{1,...,J\}$, the following lemma demonstrates the key property of the confidence interval \eqref{eq: u j confidence interval} with the scaling factor \eqref{eq: beta}.

\begin{figure}[t]

  \centering
  \centerline{\includegraphics[scale=0.21]{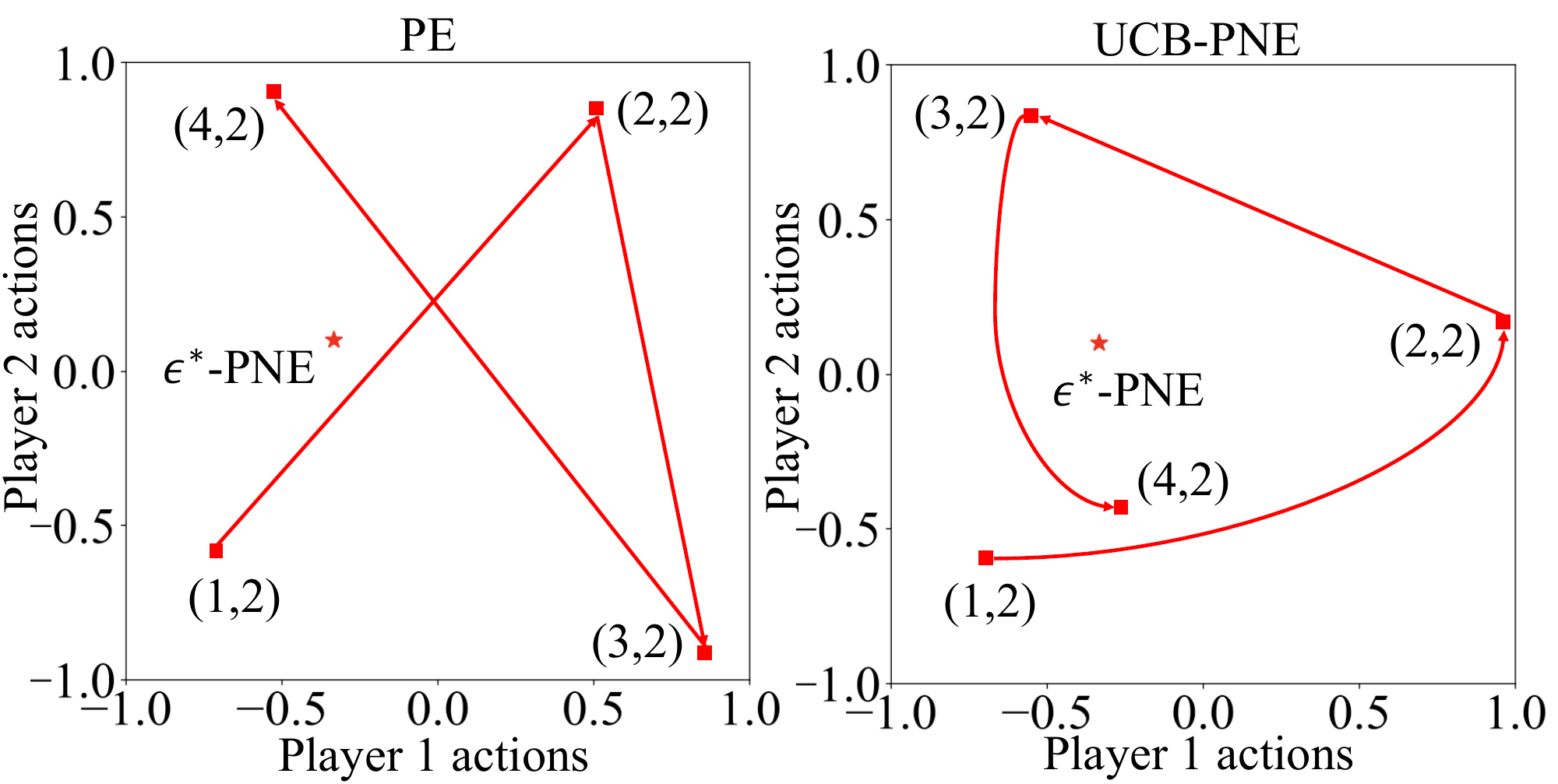}}
  \vspace{-0.2cm}
  \caption{Visualization of PE (left) and UCB-PNE (right) optimization trajectories with total query budget $\Lambda=32$. The indices $(t,m_{n,t})$ associated to each point describe the time step $t$ and the maximum-fidelity $m_{n,t}=M=2$ assigned to all players at time step $t$. The $\epsilon^*$-PNE is marked as a red star.}
  \label{fig: pe trajectory}
\vspace{-0.3cm}
\end{figure}

\begin{lemma}[Uniform Error Bound in RKHS \texorpdfstring{\cite[Theorem 2]{chowdhury2017kernelized}}{}]\label{lemma: tj rkhs utility uniform error bound}
    Let the scaling parameter $\beta_{n,j}$ be selected as in \eqref{eq: beta}. Then, under Assumptions \ref{assumption: RKHS} and \ref{assumption: rkhs norm}, with probability at least $1-\delta$ for any $\delta\in(0,1)$, the confidence interval \eqref{eq: u j confidence interval} is simultaneously valid for all inputs $\mathbf{x}\in\mathcal{X}$ and episodes $j\in\mathcal{J}$, i.e.,
    \vspace{-0.2cm}\begin{align}
        \Pr\big(&|\mu^{(M)}_{n,j,t_j}(\mathbf{x})-u_n(\mathbf{x})|\leq\beta_{n,j}\sigma^{(M)}_{n,j,t_j}(\mathbf{x}),\nonumber\\&\text{for all $\mathbf{x}\in\mathcal{X}$ and $j\in\mathcal{J}$}\big)\geq 1-\delta.\label{eq: rkhs utility uniform error bound}
    \end{align}
\end{lemma}
Lemma \ref{lemma: tj rkhs utility uniform error bound} implies the following properties for the bounds \eqref{eq: pr f cover} on the dissatisfaction of each player $n\in\mathcal{N}$.
\begin{lemma}[Dissatisfaction Coverage \texorpdfstring{\cite[Lemma 1]{tay2023no}}{}]\label{lemma: f interval}
    By building the confidence interval as in \eqref{eq: j Delta ci} for each player $n\in\mathcal{N}$, with probability at least $1-\delta$ for any $\delta\in(0,1)$, the inequalities
    \begin{align}
        \check{f}_{n,j}(\mathbf{x})\leq f_n(\mathbf{x})\leq\hat{f}_{n,j}(\mathbf{x})\label{eq: pr f cover}
    \end{align}
    and
    \begin{align}
        \hat{f}_{n,j}(\mathbf{x})-\check{f}_{n,j}(\mathbf{x})\leq 2\beta_{n,j}\Big(\sigma^{(M)}_{n,j,t_j}(\mathbf{x})+\sigma^{(M)}_{n,j,t_j}(\mathbf{x}^{(n_j)})\Big)
    \end{align}
    are simultaneously valid for all $\mathbf{x}\in\mathcal{X}$ and episodes $j\in\mathcal{J}$.
\end{lemma}

Using these results, the regret $R(\Lambda)$ of MF-UCB-PNE is upper bounded as follows.

\begin{theorem}[Regret Bound of MF-UCB-PNE]\label{theorem: regret bound}
    Under Assumptions \ref{assumption: RKHS} and \ref{assumption: rkhs norm}, with probability at least $1-N\delta$, with $\delta\in(0,1/N)$, the sequence of decision pairs selected by MF-UCB-PNE given total query budget $\Lambda$ incurs a cumulative regret bounded as
    \begin{align}
        R(\Lambda)\leq\bigg(1-\frac{\epsilon^*}{C}\bigg)\sqrt{\Lambda}\cdot\gamma_{J}+\frac{4\beta_{J}}{C}\sqrt{4(J+2)\gamma_{J}},\label{eq: regret upper bound}
    \end{align}where $\beta_{J}=\max_{n\in\mathcal{N}}\beta_{n,J}$, and $\gamma_{J}=\max_{n\in\mathcal{N}}\gamma_{n,J}$.
\end{theorem}
\begin{proof}
    Please see Appendix \ref{appendix: theorem 1 proof}.
\end{proof}

Leveraging the per-episode decomposition \eqref{eq: episode regret}, the regret upper bound \eqref{eq: regret upper bound} is the sum of a regret term incurred while exploring lower fidelity utilities and of a regret term that accounts for the dissatisfaction levels accrued during the evaluation phases. Specifically, the first term in \eqref{eq: regret upper bound} is derived by considering the largest regret obtained during the exploration phases. In contrast, the second term in \eqref{eq: regret upper bound} is akin to the regret bound for the single-fidelity multi-player case studied in \cite{tay2023no}.

The upper bound \eqref{eq: regret upper bound} can be leveraged to derive the following no-regret result.

\begin{corollary}[No-Regret of MF-UCB-PNE]\label{corollary: no regret}
    Adopting the RBF kernels \eqref{eq: rbf kernels} and \eqref{eq: fidelity kernel} in the MOGP surrogate model, under Assumptions \ref{assumption: RKHS} and \ref{assumption: rkhs norm}, the MF-UCB-PNE policy ensures the asymptotic no-regret performance \eqref{eq: no regret}.
\end{corollary}
\begin{proof}
    See Appendix \ref{appendix: corollary proof}.
\end{proof}

Finally, we note that the threshold $\eta$ in condition \eqref{eq: exploring condition 2} is not explicitly included in the regret \eqref{eq: regret upper bound}, but setting smaller $\eta$ would lead to shorter episodes with larger $J$, thus increasing the regret upper bound \eqref{eq: regret upper bound}. As verified via numerical results in the following section, this result suggests that it is preferable to choose larger values of the threshold $\eta$ when total query budget $\Lambda$ increases.

\section{Experiments}\label{sec: experiment}
In this section, we empirically evaluate the performance of the proposed MF-UCB-PNE on the synthetic random utility functions adopted in \cite{tay2023no}, as well as on applications to wireless communications systems at the physical and medium access control layers \cite{liang2007power,gkatzikis2011medium}.

\subsection{Benchmarks}\label{ssec: benchmarks}
To the best of our knowledge, MF-UCB-PNE is the first multi-fidelity optimization framework for evaluating PNEs. In light of this, the following \emph{single-fidelity} benchmarks are considered  in the experiments:
\begin{itemize}
    \item \emph{Probability of Equilibrium} (PE) \cite{picheny2019bayesian}: PE selects iterates $\mathbf{x}_t$ with the goal of maximizing the probability of obtaining a PNE, where the probability is approximated using the GP posterior at the maximum fidelity level $M$.
    \item \emph{UCB-PNE} \cite{tay2023no}: UCB-PNE applies only the evaluation phase of MF-UCB-PNE at each episode, choosing each iterate $\mathbf{x}_t$ as in \eqref{eq: x t+1}.
\end{itemize}

Being single fidelity schemes, both PE and UCB-PNE select the maximum-fidelity levels $\mathbf{m}_t=M\cdot\mathbf{1}$ at each round $t$.

\subsection{Synthetic Random Utility Functions}\label{ssec: synthetic game}
For the first experiment, we consider a game with $N=2$ players and $M=2$ fidelity levels. The $2\times 1$ multi-fidelity function vector $\mathbf{u}_n(\mathbf{x})$ for each player $n$ is sampled from the MOGP prior with lengthscale parameters in the RBF kernels \eqref{eq: rbf kernels} and \eqref{eq: fidelity kernel} given by $h=0.89$ and $\zeta^{(1)}=\zeta^{(2)}=0.78$, respectively. The correlation parameter in the auto-regressive model \eqref{eq: ar model} is set as $\rho^{(1)}=0.768$, and the observations noise variance is $\sigma^2=0.1$. The action space for each player is restricted to the interval $\mathcal{X}_n=[-1,1]$, and the cost levels are $\lambda^{(1)}=1$ and $\lambda^{(2)}=8$. For all schemes, we consider choices both a well-specified MOGP, which adopts the same parameters used in the random functions generation process, and a misspecified MOGP, which sets different parameters $h=0.62,\zeta^{(1)}=\zeta^{(2)}=0.41$ and  $\rho^{(1)}=0.625$. To evaluate the $\epsilon$-PNE of the synthetic game, we uniformly discretize the continuous action space $\mathcal{X}_n$ into $128$ strategies for every player, and we find a PNE by a brute-force grid search.

To start, for a well-specified MOGP, Fig. \ref{fig: trajectory} shows the trajectories of the decisions made by MF-UCB-PNE with total budget $\Lambda=32$ and threshold $\eta=0.5$. The two figures show a heat map of the utility functions for the two players. For these utility functions, the smallest achievable dissatisfaction level is $\epsilon^*=0.19$, and there is a unique $\epsilon^*$-PNE, which is shown as a red star. The decisions made during the exploration phases are denoted by black circles, while the actions selected during the evaluation phases are represented by red squares. For comparison, Fig. \ref{fig: pe trajectory} shows the evolution of the iterates of the benchmarks PE and UCB-PNE for the same budget $\Lambda=32$. As seen in Fig. \ref{fig: pe trajectory}, by always targeting high-fidelity solutions, the benchmark methods run out of budget early, obtaining solutions that are far from the $\epsilon^*$-PNE.

As observed in Fig. \ref{fig: trajectory}, the iterates during the exploration phase of MF-UCB-PNE tend to be spread out over the entire action space. This validates the role of the acquisition function \eqref{eq: select exploring decision} enabling exploration by reducing the uncertainty about the utility function $u_n(\mathbf{x})$. In contrast, the iterates during the evaluation phase tend to concentrate on areas in which both utility functions are estimated to be large, thus moving close to an $\epsilon^*$-PNE solution.

\begin{figure}[t!]
     \centering
     \hspace{-0.3cm}
     \subfigure[]{
         \centering
         \includegraphics[scale=0.157]{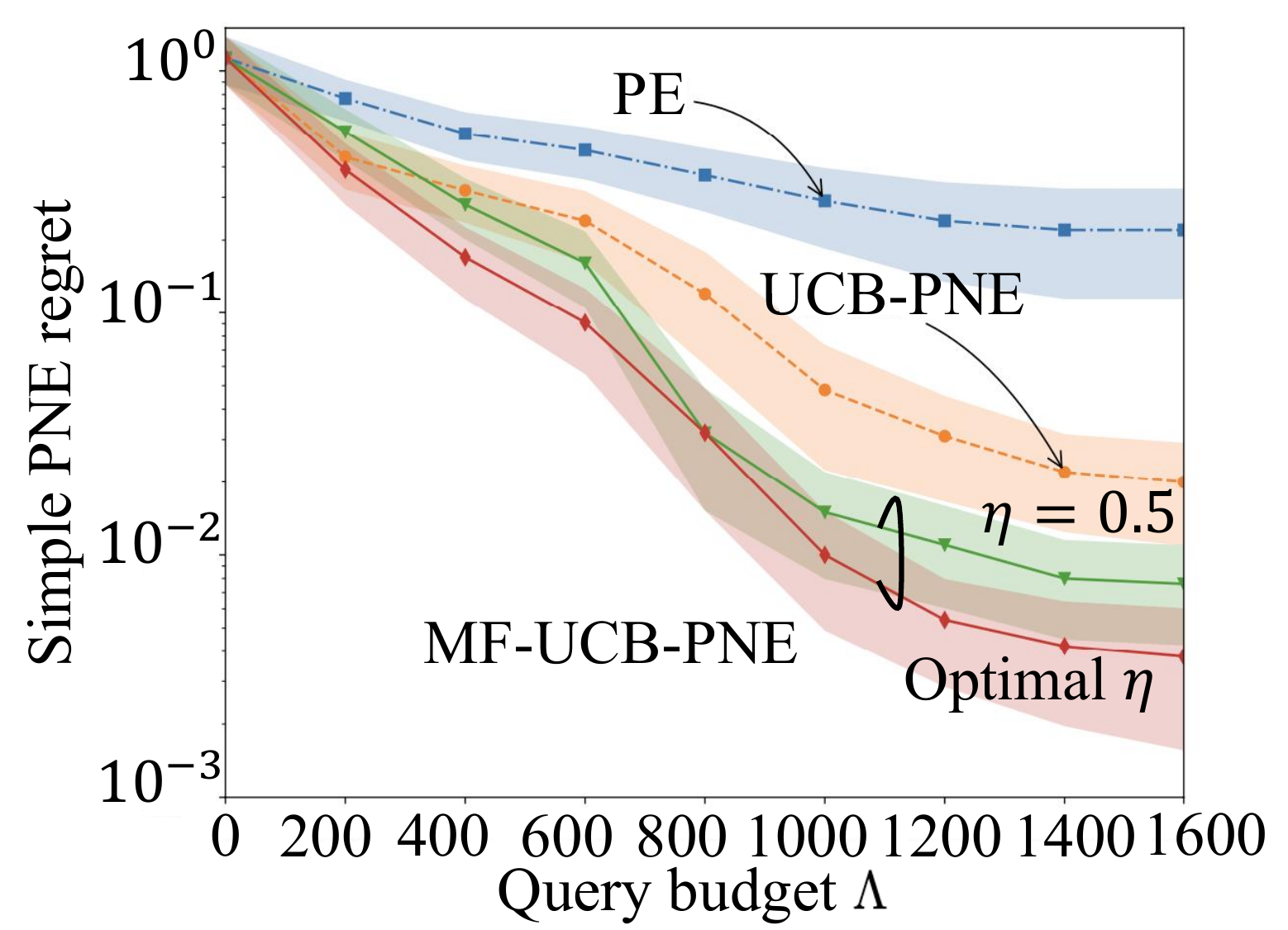}
         \label{fig:y equals x}
     }
     \hspace{-0.5cm}
     \subfigure[]{
         \centering
         \includegraphics[scale=0.162]{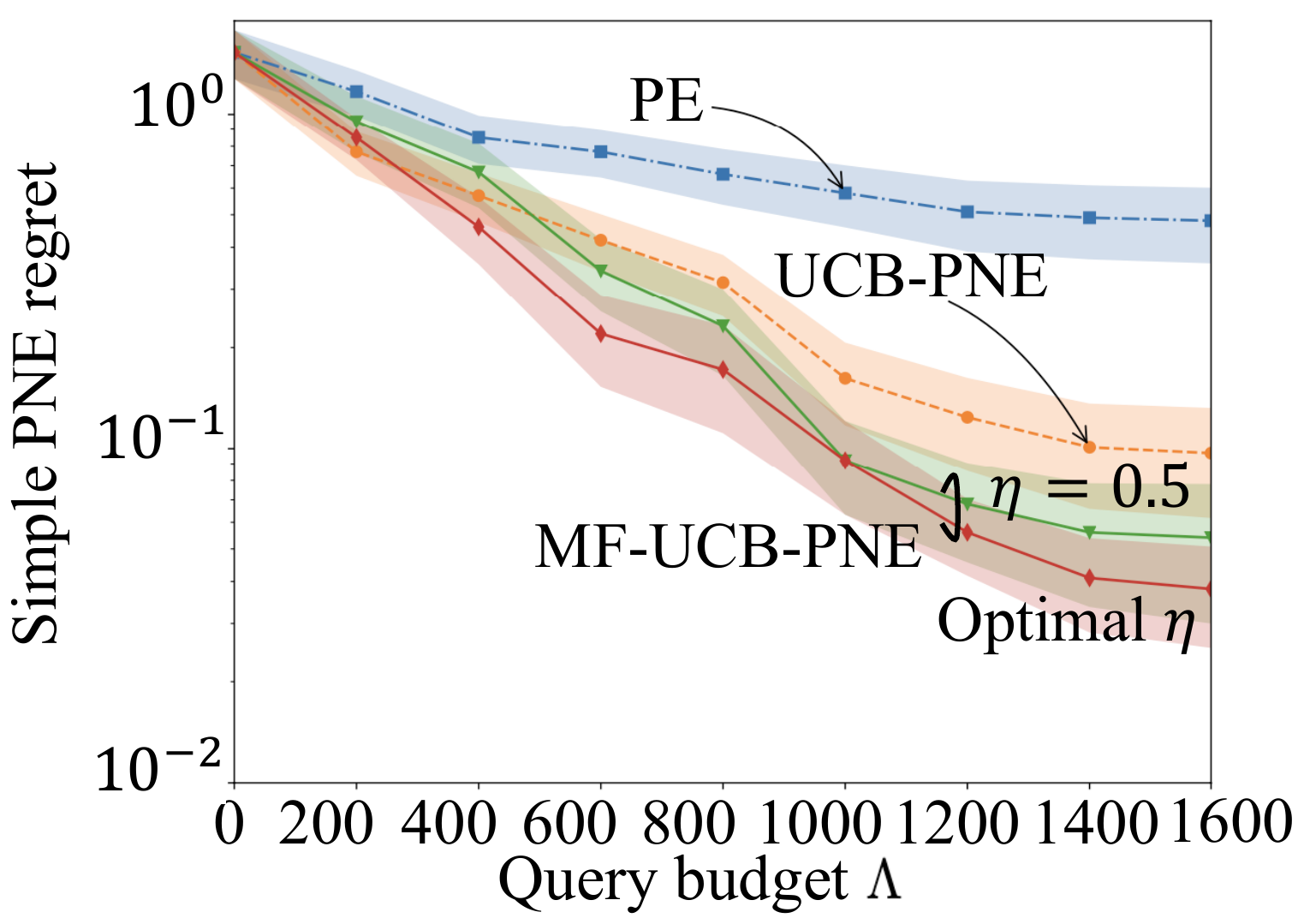}
         \label{fig:three sin x}
     }
     \vspace{-0.2cm}
        \caption{Simple pure Nash regret \eqref{eq: simple regret} against the total query budget $\Lambda$ for a well-specified MOGP (a) and for a misspecified MOGP (b), attained by PE (blue dash-dotted line), UCB-PNE (orange dashed line), MF-UCB-PNE with $\eta=0.5$ (green solid line) and MF-UCB-PNE with optimal $\eta$ at different budget (red solid line). }
  \label{fig: regret vs budget}
  \vspace{-0.3cm}
\end{figure}

To corroborate the main conclusion obtained from the observations of Fig. \ref{fig: trajectory} and Fig. \ref{fig: pe trajectory}, as in \cite{tay2023no} we analyze the \emph{simple PNE regret} 
\begin{align}
    R_J=\max\limits_{n\in\mathcal{N},\mathbf{x}\in\{\mathcal{E}_j\}_{j=1}^J}\{f_n(\mathbf{x})\}-\epsilon^*,\label{eq: simple regret}
\end{align}
which evaluates the gap between the best dissatisfaction level attained throughout the optimization process and the minimum value $\epsilon^*$. Unless stated otherwise, all the results are averaged over $50$ experiments, reporting $90\%$ confidence level.

We now increase the number of players to $N=10$, as well as the discrete fidelity space to $M=4$ levels with corresponding query costs $\lambda^{(1)}=1, \lambda^{(2)}=2, \lambda^{(3)}=4$, and $\lambda^{(4)}=8$, and we plot the simple PNE regret \eqref{eq: simple regret} as a function of total query budget $\Lambda$ in Fig. \ref{fig: regret vs budget}, while Fig. \ref{fig: regret vs eta} considers also the dependence on the threshold $\eta$. For MF-UCB-PNE, in Fig. \ref{fig: regret vs budget}, we consider a fixed threshold $\eta=0.5$ in the stopping condition \eqref{eq: exploring condition 2}, as well as the optimal threshold $\eta$ obtained by a discrete exhaustive search.

For a well-specified MOGP, as seen in Fig. \ref{fig:y equals x} and Fig. \ref{fig: regret vs eta}, MF-UCB-PNE with optimal threshold $\eta$ obtains the best regret performance across all query budget settings. As shown in Fig. \ref{fig: regret vs eta}, the optimal value of threshold $\eta$ increase with the budget $\Lambda$, indicating that longer exploration phases are beneficial when the budget constraints are less strict. As illustrated in Fig. \ref{fig:three sin x}, for a misspecified MOGP case,  the simple PNE regret of all schemes are significantly increased due to the inaccurate surrogate modeling adopted in the decisions acquisition process. In this case, for very limited query budgets, here $\Lambda\leq400$, MF-UCB-PNE may be outperformed by UCB-PNE owing to its reliance on a more complex MOGP. However, even in this setting, when the budget $\Lambda$ is large enough, MF-UCB-PNE outperforms all other schemes with an optimized choice of the parameter $\eta$. 

\begin{figure}[t]

  \centering
  \centerline{\includegraphics[scale=0.28]{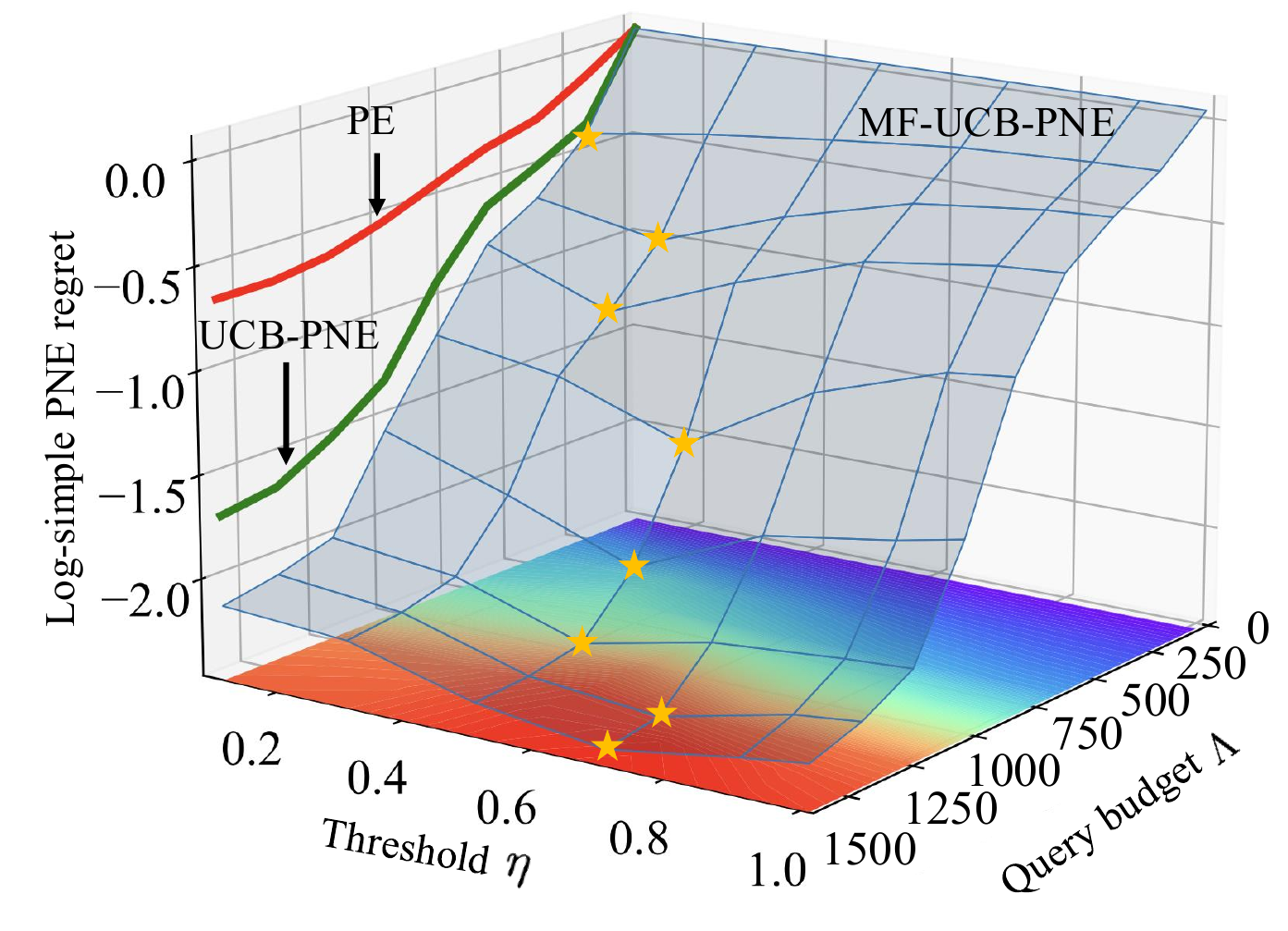}}
  \vspace{-0.3cm}
  \caption{Log-simple PNE regret \eqref{eq: simple regret} against the threshold $\eta$ in \eqref{eq: exploring condition 2} and the total query budget $\Lambda$, for PE (red solid line), UCB-PNE (green solid line), and MF-UCB-PNE (blue surface). The optimal values of the threshold $\eta$ given the corresponding total query budget $\Lambda$ are labeled as yellow stars.}
  \label{fig: regret vs eta}
\vspace{-0.3cm}
\end{figure}

\subsection{Optimizing Power Allocation in Wireless Systems}\label{ssec: uplink power}
In this section, we study an application to the physical layer of a wireless system. In this setting, $N$ interfering transmission links, which act as players, are assigned transmission powers $\mathbf{x}$ by a centralized coordinator. Each $n$-th link is characterized by a direct channel gain $h_{n,n}$, while interference between the transmitter of the $n$-th link and the receiver of the $m$-th link is described by a channel $h_{n,m}$ with $n\neq m$. The channels $\{h_{n,m}\}_{n,m\in\mathcal{N}}$ are random and have unknown distributions.

The utility function $u_n(\mathbf{x})$ for the $n$-th link is given by the average spectral efficiency at which the link can communicate with a penalty term that increases with the transmission power $x_n$ \cite{liang2007power}, i.e.,
\begin{align}
    u_n(\mathbf{x})&=\mathbbm{E}[\log(1+\text{SINR}_n)]-\xi_nx_n\nonumber\\&=\mathbbm{E}\Bigg[\log\bigg(1+\frac{|h_{n,n}|^2x_n}{\sigma^2_z+\sum_{m=1,m\neq n}^N|h_{m,n}|^2x_{m}}\bigg)\Bigg]\nonumber\\&\hspace{0.4cm}-\xi_nx_n\quad\text{[bps/Hz]},\label{eq: power control utility}
\end{align}
where $\sigma^2_z$ is the noise power and $\xi_n$ is a fixed factor dictating the relative importance of power consumption and transmission rate for link $n$. The spectral efficiency utility \eqref{eq: power control utility} is averaged over the unknown distributions of the channels.

\begin{figure}[t]

  \centering
  \centerline{\includegraphics[scale=0.28]{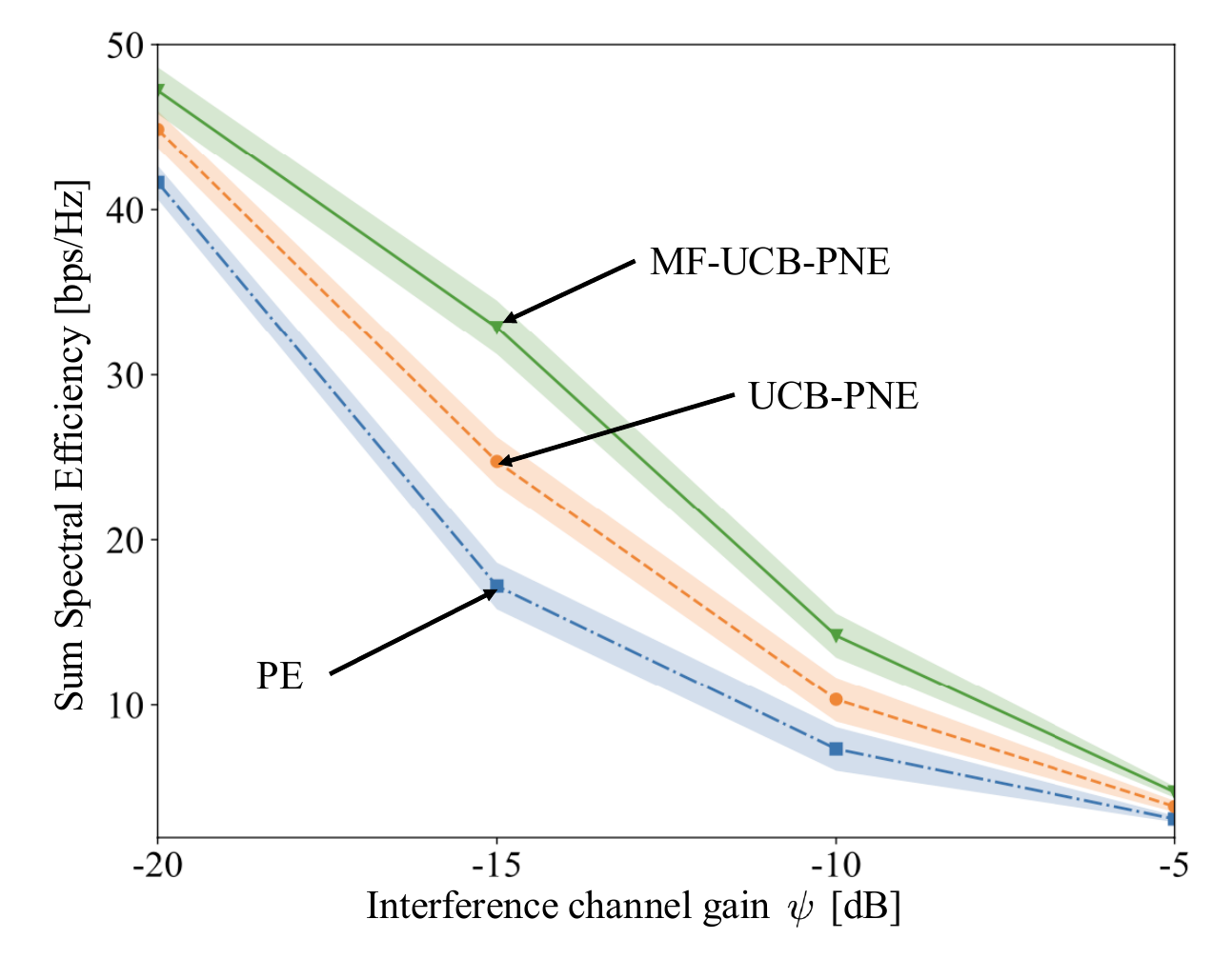}}
  \vspace{-0.4cm}
  \caption{Sum spectral efficiency against the interference channel gain $\psi$ for PE (blue dash-dotted line), UCB-PNE (orange dashed line), and MF-UCB-PNE with optimal threshold (green solid line). The total query budget is set to $\Lambda=36000$.}
  \label{fig: sse vs inter}
\vspace{-0.3cm}
\end{figure}



The centralized optimizer can estimate the utilities \eqref{eq: power control utility} by observing $\lambda$ realizations of the random variable $\log(1+\text{SINR}_n)$. The cost is given by the number of realizations, and we allow for the options $\lambda^{(1)}=1, \lambda^{(2)}=10, \lambda^{(3)}=20, \lambda^{(4)}=50$ and $\lambda^{(5)}=100$. We assume $N=20$ links, each with normalized transmission power $-13$ dB $\leq x_n\leq 23$ dB, and noise power $\sigma_z^2=-20$ dB. We simulate Rayleigh fading direct channels $h_{n,n}\sim\mathcal{CN}(0,1)$ and interference channels $h_{n,m}\sim\mathcal{CN}(0,\psi)$ with $\psi=-20$ dB. The MOGP is initialized with parameters $h=0.89,\zeta^{(m)}=0.78$, and $\rho=0.768$. Note that the MOGP is not well-specified, since data generation does not follow an MOGP. For MF-UCB-PNE, we consider a fixed threshold $\eta=0.4$, as well as the optimal $\eta$ obtained in a similar manner to Fig. \ref{fig: regret vs budget}.


Fig.  \ref{fig: sse vs inter} shows the largest average sum-spectral efficiency attained during the optimization procedure as a function of the interference channel gain $\psi$ for a penalty coefficient $\xi_n=0.1$ for all players. At low interference levels, the utilities of the links are approximately decoupled, and all schemes attain similar performance. The performance gap between MF-UCB-PNE and the other benchmarks is evident in the regime of moderate interference gains, while the performance of all schemes degrade at higher interference levels, causing all transmitters to leverage the maximum transmission power.

\subsection{Optimizing a Random Access Protocol}\label{ssec: mac game}
Finally, we consider an optimization problem operating at the medium access layer for an ALOHA random access protocol. For each mobile terminal $n$, the central optimizer must select probabilities $\mathbf{x}_n=(x_{n,1},x_{n,2})$, with $x_{n,1}\in[0,1]$ describing the probability of mobile terminal $n$ being active and $x_{n,2}\in[0,1]$ denoting the probability of mobile terminal $n$ attempting a channel access \cite{gkatzikis2011medium}. Assuming a standard collision channel, the probability of a successful transmission for terminal $n$ is
\begin{align}
    T_n(\mathbf{x})=\frac{x_{n,1}x_{n,2}}{1-x_{n,1}x_{n,2}}\prod_{n'\in\mathcal{N}}(1-x_{n',1}x_{n',2}),\label{eq: mac utility}
\end{align}
while the average energy consumption is
\begin{align}
    E_n(\mathbf{x})=x_{n,1}(c_1+c_2x_{n,2}),\label{eq: energy consumption}
\end{align}
where $c_1$ and $c_2$ represent the energy consumption to be active and to transmit, respectively. The utility function is $u_n(\mathbf{x})=T_n(\mathbf{x})-\xi_nE_n(\mathbf{x})$, where $\xi_n>0$ represents the relative importance of the expected energy consumption $E_n(\mathbf{x})$ and throughput $T_n(\mathbf{x})$.

\begin{figure}[t]

  \centering
  \centerline{\includegraphics[scale=0.28]{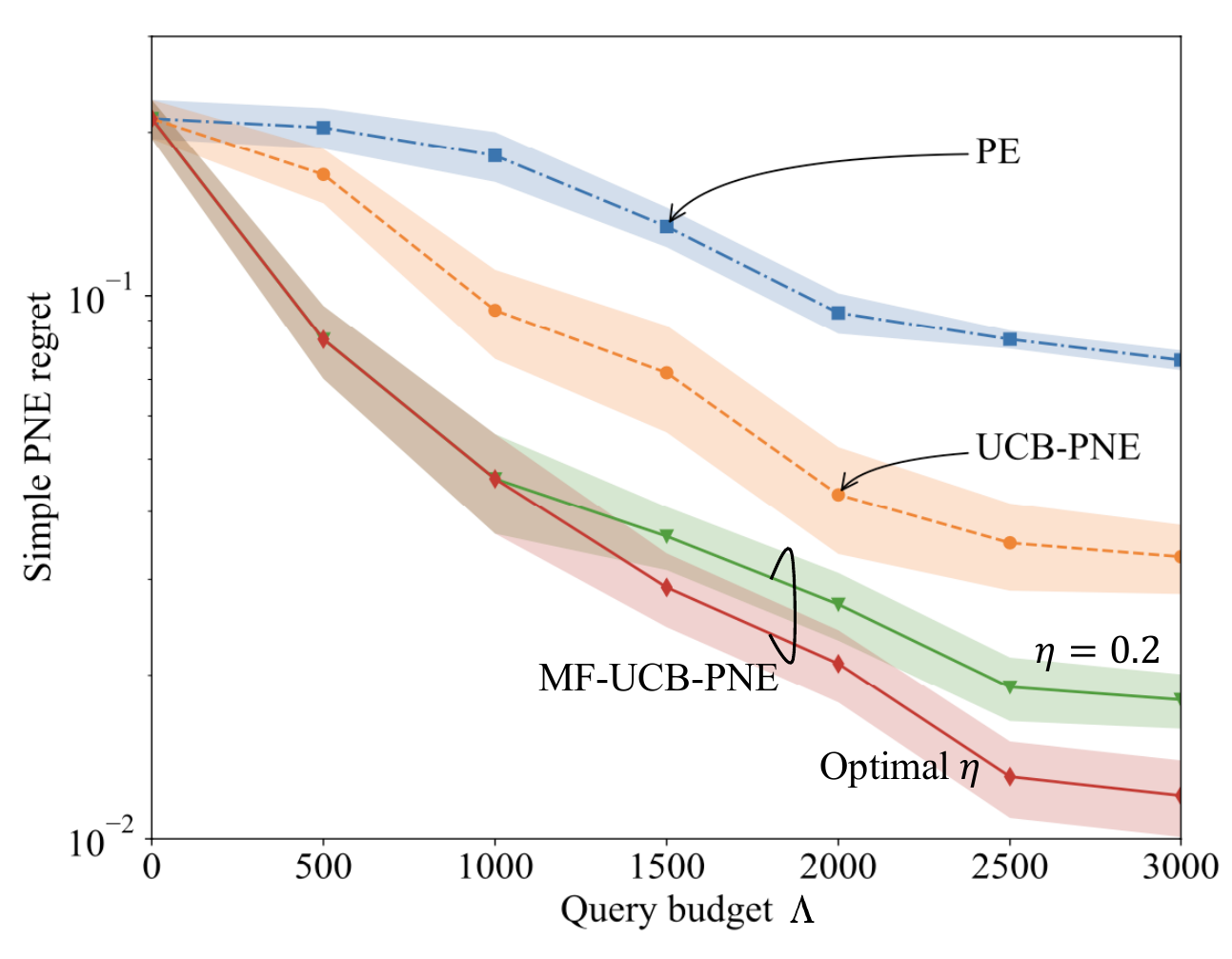}}
  \vspace{-0.4cm}
  \caption{Simple PNE regret \eqref{eq: simple regret} against the total query budget $\Lambda$ for PE (blue dash-dotted line), UCB-PNE (orange dashed line), and MF-UCB-PNE with threshold $\eta=0.2$ (green solid line) and optimal threshold $\eta$ (red solid line).}
  \label{fig: mac regret vs budget}
\vspace{-0.3cm}
\end{figure}

We consider a scenario of $N=5$ mobile terminals with energy consumptions $c_1=50$ units for staying active and $c_2=70$ units for transmission, while the maximum energy consumption constraints for the mobile terminals are set to $\{60, 55, 50, 45, 40\}$. The MOGP is initialized with parameters $h=1.08,\zeta^{(m)}=0.41,\rho=0.797$ for all players. And we set a threshold $\eta=0.2$ as well as the optimal $\eta$ for MF-UCB-PNE. We set $\xi_n=6.5\times 10^{-4}$ for all players, and low-fidelity evaluations of the utility function use the factors $\omega^{(1)}=4.9\times10^{-4}, \omega^{(2)}=5.5\times10^{-4}$, and $\omega^{(3)}=6.1\times10^{-4}$ instead, with the corresponding query costs $\lambda^{(1)}=1,\lambda^{(2)}=5,\lambda^{(3)}=10$ and $\lambda^{(4)}=20$. The equilibrium solution is computed by fixed-point iterations \cite{von2012newton}.

Fig. \ref{fig: mac regret vs budget} shows the simple PNE regret \eqref{eq: simple regret} as a function of the total query budget $\Lambda$. MF-UCB-PNE attains the best performance over the other benchmarks across all values of the total query budget. For example, with $\Lambda=3000$, MF-UCB-PNE achieves throughput levels $\{T_n\}_{n\in\mathcal{N}}=\{0.09187, 0.07232, 0.06236, 0.05229, 0.04627\}$ at energy levels $\{E_n\}_{n\in\mathcal{N}}=\{57, 55, 50, 45, 40\}$, while UCB attains throughput levels $\{T_n\}_{n\in\mathcal{N}}=\{0.07059, 0.06785, 0.06109, 0.05446, 0.04795\}$ with the maximum energy consumption. 

Moreover, MF-UCB-PNE exhibits an earlier and steeper decline in regret over low to moderate budgets, indicating that lower fidelity exploration has already localized promising regions before substantial high-fidelity queries with high costs. Near the terminal budget, MF-UCB-PNE shows a final sharp contraction as the remaining budget is concentrated on high-fidelity checks around the best solutions obtained so far, whereas UCB-PNE and PE tend to flatten after uniform high-fidelity sampling. This early localization trajectory followed by targeted certification aligns with the two-phase design and explains the lower regret achieved under the same total budget.

\section{Conclusion}\label{sec: conclusion}
The paper has introduced a centralized optimization framework for identifying approximate pure Nash equilibria (PNE) in complex multi-agent systems with black-box utility functions. The framework was devised within a learned optimization problem with an equilibrium constraint in which a central controller coordinates multiple strategic players. To tackle the challenge of learning unknown utilities through sequential queries, the paper presented MF-UCB-PNE, a novel multi-fidelity Bayesian optimization algorithm. MF-UCB-PNE explicitly balances exploration and exploitation across fidelity levels under stringent query cost constraints. By leveraging inexpensive low-fidelity evaluations for broad exploration and transitioning to high-fidelity queries for refined exploitation, the algorithm efficiently converges toward configurations that approximate PNE solutions while respecting budget limits. The paper provides a rigorous theoretical regret analysis, proving that MF-UCB-PNE achieves asymptotic no-regret performance under suitable conditions. Empirically, results show that the proposed approach consistently identifies high-quality equilibrium configurations within limited query budgets. 

Looking ahead, several promising directions for future research emerge from this work. One avenue is to generalize the framework beyond pure Nash equilibria to broader equilibrium concepts, such as mixed or correlated equilibria \cite{vlatakis2020no}. Extending the multi-fidelity approach to these scenarios could accommodate games where randomness or coordination among players yields more stable outcomes.  Furthermore, enhancing the scalability of the proposed method is important for applications to large-scale multi-agent systems involving more complex interaction structures and higher-dimensional strategy spaces \cite{maddox2021bayesian}.

\appendix

\subsection{Proof of Theorem \ref{theorem: regret bound}}\label{appendix: theorem 1 proof}
We start the proof of Theorem \ref{theorem: regret bound} by deriving an upper bound for the episode regret \eqref{eq: episode regret} as a function of the allocated query budget $\Lambda_j$. By Definitions \ref{definition: episode} and \ref{definition: episode regret}, we have the inequality
\begin{align}
    R(\mathcal{E}_j)=\frac{\Lambda_j}{N}\Big(1-\frac{\epsilon^*}{C}\Big)-\Big(1-\frac{\epsilon_j}{C}\Big).\label{eq: expansion of episode regret}
\end{align}
Assuming that MF-UCB-PNE terminates within $J$ episodes, by Definition \ref{definition: episode regret}, the regret upper bound can be then obtained as
\begin{align}
    R(\Lambda)&=\frac{\Lambda}{N}\Big(1-\frac{\epsilon^*}{C}\Big)-\sum_{j=1}^J\Big(1-\frac{\epsilon_j}{C}\Big)\tag{55a}\label{eq: regret line 1}\\&=\Big(\frac{\Lambda}{N}-J\Big)\Big(1-\frac{\epsilon^*}{C}\Big)+\sum_{j=1}^J\frac{\epsilon_j-\epsilon^*}{C}\tag{55b}\label{eq: regret line 3}\\&=\Big(\frac{1-\epsilon^*/C}{N}\Big)\sum_{j=1}^J\Lambda_{\mathcal{L}_{j}}+\sum_{j=1}^J\frac{\epsilon_j-\epsilon^*}{C}.\tag{55c}\label{eq: expanded cumulative regret}
\end{align}
The inequality \eqref{eq: regret line 1} is obtained by summing the episode regret in \eqref{eq: expansion of episode regret}; equation \eqref{eq: regret line 3} subtracts the best achievable dissatisfaction $J\epsilon^*$ from \eqref{eq: regret line 1} across $J$ episodes, and applies the associative law.
The first term in \eqref{eq: expanded cumulative regret} represents the regret incurred from decision sequences in exploration phases, while the second term in \eqref{eq: expanded cumulative regret} describes the regret obtained by optimizing the utility functions at the maximum fidelity level.

As shown in Algorithm \ref{table: lf-mi-max}, the exploration phase termination conditions \eqref{eq: exploring condition 2} and \eqref{eq: exploring condition 3} are checked by the optimizer sequentially at each time step before including the decision pair \eqref{eq: select exploring decision} in the sequence $\mathcal{L}_{j}$. Therefore, we can guarantee that the cumulative information gain-cost ratio induced on sequence $\mathcal{L}_{j}$ always satisfies the inequality
\begin{align}\tag{56}
    \frac{\sum_{n=1}^N\mathrm{I}(\mathbf{y}_{n,j,t_j};\{u_n(\mathbf{x})\}_{\mathbf{x}\in\mathcal{L}_{j}}|\mathcal{L}_{j},\mathcal{D}_{n,j,0})}{\Lambda_{\mathcal{L}_{j}}}\geq\frac{1}{\sqrt{\Lambda^{\text{tot}}_j}},\label{eq: stopping condition}
\end{align}
where the $t_j\times 1$ vector $\mathbf{y}_{n,j,t_j}$ includes the utility observations obtained on sequence $\mathcal{L}_{j}$ performed by player $n$. Therefore, the cumulated query cost in \eqref{eq: expanded cumulative regret} can be upper bounded as
\begin{align}
    \sum_{j=1}^J\Lambda_{\mathcal{L}_{j}}&\leq\sum_{j=1}^J\frac{\sum_{n=1}^N\mathrm{I}(\mathbf{y}_{n,j,t_j};\{u_n(\mathbf{x})\}_{\mathbf{x}\in\mathcal{L}_{j}}|\mathcal{L}_{j},\mathcal{D}_{n,j,0})}{\frac{1}{\sqrt{\Lambda^{\text{tot}}_j}}}\nonumber\\&\leq N\sqrt{\Lambda}\gamma_{J},\tag{57}\label{eq: low fidelity regret upper bound}
\end{align}
with the maximal information gain $\gamma_{J}=\max_{n\in\mathcal{N}}\gamma_{n,J}$. 

Let us denote $\beta_{J}=\max_{n\in\mathcal{N}}\beta_{n,J}$ as the maximum scaling parameter among all players. Then the upper bound for the second term in \eqref{eq: expanded cumulative regret} can be obtained as
\begin{align}
    &\sum_{j=1}^J(\epsilon_j-\epsilon)\nonumber\\&=\sum_{j=1}^J\Big[\max_{n\in\mathcal{N}}f_n(\mathbf{x}_{j,t_j+1})-\max_{n'\in\mathcal{N}}f_{n'}(\mathbf{x}^*)\Big]\tag{58a}\\&\leq\sum_{j=1}^J\Big[\max_{n\in\mathcal{N}}\hat{f}_{n,j}(\mathbf{x}_{j,t_j+1})-\max_{n'\in\mathcal{N}}\check{f}_{n',j}(\mathbf{x}_{j,t_j+1})\Big]\tag{58b}\label{eq: epsilon diff line 4}\\&\leq\sum_{j=1}^J\Big[\hat{f}_{n_j,j}(\mathbf{x}_{j,t_j+1})-\check{f}_{n_j,j}(\mathbf{x}_{j,t_j+1})\Big]\tag{58c}\label{eq: epsilon diff line 5}\\&\leq\sum_{j=1}^J2\beta_{n,j}\Big(\sigma^{(M)}_{n,j,t_j}(\mathbf{x}_{j,t_j+1})+\sigma^{(M)}_{n,j,t_j}(\mathbf{x}^{(n_j)})\Big)\tag{58d}\label{eq: from lemma 1}\\&\leq 4\beta_{J}\Bigg(\sum_{j=1}^J\max\Big(\sigma^{(M)}_{n,j,t_j}(\mathbf{x}_{j,t_j+1}),\sigma^{(M)}_{n,j,t_j}(\mathbf{x}^{(n_j)})\Big)\Bigg)\tag{58e}\\&\leq 4\beta_{J}\sqrt{4(J+2)\gamma_{J}},\tag{58f}\label{eq: sfbo upper bound}
\end{align}
where the inequality \eqref{eq: epsilon diff line 4} is obtained by introducing the upper and lower confidence intervals obtained on the selected action profile $\mathbf{x}_{j,t_j+1}$ instead of the $\epsilon^*$-PNE solution $\mathbf{x}^*$; \eqref{eq: epsilon diff line 5} is obtained by aligning the player index which has the maximum dissatisfaction upper confidence bound as in \eqref{eq: worst player}; \eqref{eq: from lemma 1} is obtained by applying the union bound to Lemma \ref{lemma: f interval}; and the inequality \eqref{eq: sfbo upper bound} is obtained from \cite[Lemma 4]{chowdhury2017kernelized}. 

Therefore, combining \eqref{eq: low fidelity regret upper bound} and \eqref{eq: sfbo upper bound}, we get the upper bound for the regret $R(\Lambda)$ as in \eqref{eq: regret upper bound}, concluding the proof.

\subsection{Proof of Corollary \ref{corollary: no regret}}\label{appendix: corollary proof}
We start the proof of Corollary \ref{corollary: no regret} by showing that the maximal information gain terms in Theorem \ref{theorem: regret bound} grows sub-linearly with the query budget $\Lambda$. According to the selection subroutine in Algorithm \ref{table: lf-mi-max}, the maximal mutual information at the first time step is obtained at the highest fidelity level, i.e.,
\begin{align}
    &\max\limits_{\substack{\mathbf{x}\in\mathcal{X}\\m_n\in\mathcal{M}}}\{\mathrm{I}(y_n^{(m_n)};u_n(\mathbf{x})|\mathbf{x},m_n)=\frac{1}{2}\ln(1+\frac{1}{1+[\sigma_{n,0}^{(m_n)}(\mathbf{x})]^2})\}\nonumber\\&=\max\limits_{\mathbf{x}\in\mathcal{X}}\{\mathrm{I}(y_n^{(M)};u_n(\mathbf{x})|\mathbf{x},M)\}.\tag{59}\label{eq: first step mi}
\end{align}
When the fidelity decision at each time step is selected as $M$, we only perform the evaluation phase. For brevity of notation, we omit the episode index and let $\mathbf{M}_{n,t-1}=M\cdot\mathbf{1}$ denote the fidelity decisions for all the previous $t-1$ time steps in the optimization process. The maximal mutual information in \eqref{eq: first step mi} can be then inductively expressed as
\begin{align}
    \max\limits_{\substack{\mathbf{x}\in\mathcal{X}\\m_n\in\mathcal{M}}}\Bigg\{&\mathrm{I}(y_n^{(m_n)};u_n(\mathbf{x})|\mathbf{x},m_n,\mathbf{X}_{t-1},\mathbf{M}_{n,t-1})\nonumber\\&=\frac{1}{2}\ln\bigg(1+\frac{[\sigma_{n,t-1}^{(M)}(\mathbf{x})]^2}{[\sigma_{n,t-1}^{(M)}(\mathbf{x})]^2+[\sigma_{n,t-1}^{(m_n)}(\mathbf{x})]^2}\bigg)\Bigg\},\tag{60}\label{eq: inductive mi}
\end{align}
which indicates that the maximal mutual information \eqref{eq: inductive mi} is obtained by always querying the action profiles at the maximum-fidelity level $M$. Since the maximal information gain $\gamma_{n,t}$ satisfies submodularity, a diminishing returns property \cite{krause2005near}, it can be upper bounded via a greedy procedure at each time \cite[Section V.A]{srinivas2012information}, i.e.,
\begin{align}
    \gamma_{n,t}&=\max\limits_{\substack{\mathbf{X}_{t}\subset\mathcal{X}\\\mathbf{M}_{n,t}\subset\mathcal{M}}}\mathrm{I}(\mathbf{y}_{t};\{u_n(\mathbf{x})\}_{\mathbf{x}\in\mathbf{X}_t}|\mathbf{X}_{t},\mathbf{M}_{n,t})\nonumber\\&\leq\frac{1}{2}\bigg(1-\frac{1}{e}\bigg)^{-1}\sum_{t'=1}^{t}\max\limits_{\substack{\mathbf{x}_{t'}\in\mathcal{X}\\m_{n,t'}\in\mathcal{M}}} \mathrm{I}(y_{n,t'}^{(m_{n,t'})};u_n(\mathbf{x}_{t'})\nonumber\\&\hspace{0.5cm}|\mathbf{x}_{t'},m_{n,t'},\mathbf{X}_{t'-1},\mathbf{M}_{n,t'-1}).\tag{61}\label{eq: chain rule}
\end{align}
The inequality \eqref{eq: chain rule}, obtained by applying the chain rule of mutual information to \eqref{eq: inductive mi}, shows that the maximal information gain $\gamma_{n,t}$ is accrued by performing greedy mutual information maximization over action profile $\mathbf{x}_{t'}$ with highest fidelity level $m_{n,t'}=M$ at each time step $t'\geq 1$.

Given the selected RBF kernel in MOGP, the maximal information gain is upper bounded by $\gamma_{n,t}=\mathcal{O}(\log (t)^{d+1})$ \cite[Theorem 5]{srinivas2012information}

Therefore, denoting the overall number of time steps as $T_J=\sum_{j=1}^J(t_j+1)$ and plugging the maximal information gain upper bound into the regret upper bound \eqref{eq: regret upper bound}, we get
\begin{align}
    R(\Lambda)&\leq\bigg(1-\frac{\epsilon^*}{C}\bigg)\sqrt{\Lambda}\cdot\gamma_{J}+\frac{4\beta_{J}}{C}\sqrt{4(J+2)\gamma_{J}}\tag{62a}\\&\leq \mathcal{O}(\sqrt{\Lambda}\log(T_J)^{d+1})+\mathcal{O}(\sqrt{J}\log(T_J)^{d+1})\tag{62b}\label{eq: no regret line 2}\\&\leq\mathcal{O}\bigg((\sqrt{J}+\sqrt{\Lambda})\log\Big(\frac{\Lambda-J\lambda^{(M)}}{\lambda^{(1)}}+J\Big)^{d+1}\bigg),\tag{62c}\label{eq: final O}
\end{align}
where inequality \eqref{eq: no regret line 2} is obtained by using $\beta_J=\mathcal{O}(\sqrt{\gamma_J})$, and applying the maximal information gain upper bound; and inequality \eqref{eq: final O} is obtained by considering the maximum number of steps one can query given budget $\Lambda$.
The upper bound \eqref{eq: final O} satisfies the limit \eqref{eq: no regret}, hence, MF-UCB-PNE incurs asymptotic no-regret performance.


%

 


\bibliographystyle{ieeetr}
\bibliography{refer}

\vspace{11pt}


\vfill

\end{document}